\documentclass[12pt,draftcls,onecolumn]{IEEEtran} 
\IEEEoverridecommandlockouts                              % This command is only
\usepackage{amsthm}
\usepackage{graphicx}
\usepackage{graphics}
\usepackage{amssymb,mathrsfs}
\usepackage{amsmath,mathtools,eufrak}
\usepackage{color}
\usepackage{xspace}
\usepackage{algpseudocode}
\usepackage{bbm}
\usepackage{comment}
\usepackage{algorithmicx}
\usepackage{subfig}
\usepackage{psfrag}
\usepackage{tikz}
\usetikzlibrary{shapes,arrows}
\usetikzlibrary{positioning}
\usetikzlibrary{calc}
\usepackage[utf8]{inputenc}
\usepackage[english]{babel}
 
\tikzstyle{block}=[draw opacity=0.7,line width=1.4cm]

\definecolor{CranJ}{cmyk}{0,0.69,0.54,0.04} %cranberry jello
\definecolor{PinkJ}{cmyk}{0,0.71,0.43,0.12} %pink jeep
\definecolor{Cran}{cmyk}{0,0.73,0.41,0.29} %cranberry 
\definecolor{VRed}{cmyk}{0,0.75,0.25,0.2} %violetred
\definecolor{ORed}{cmyk}{0,0.75,0.75,0} %orangered4
\definecolor{CBlue}{cmyk}{1,0.25,0,0} %curacao	
\title{\LARGE \bf
 Cycle flow formulation of optimal network flow problems for centralized and decentralized solvers}

\author{Reza Asadi and  Solmaz S. Kia  %
  \thanks{The first author is a Ph.D. candidate with the Computer Science Department, the second author is an assistant professor with the Mechanical and Aerospace Engineering Department of University of California Irvine, Irvine, CA 92697 
    {\tt\small \{rasadi, solmaz\}@uci.edu}}%
}

\newcommand{\VV}{\mathcal{V}}
\newcommand{\EE}{\mathcal{E}}
\newcommand{\GG}{\mathcal{G}}

\newcommand{\real}{{\mathbb{R}}} \newcommand{\reals}{{\mathbb{R}}}
 
\newcommand{\realpositive}{{\mathbb{R}}_{>0}}

\newcommand{\realnonnegative}{{\mathbb{R}}_{\ge 0}}

\newcommand{\argmin}{\operatorname{argmin}}
\newcommand{\rank}{\operatorname{rank}}

\newcommand{\until}[1]{\in\{1,\dots,#1\}}

\newcommand{\vect}[1]{\boldsymbol{\mathbf{#1}}}
\newcommand{\vectsf}[1]{\vect{\mathsf{#1}}}
\newcommand{\Bvect}[1]{\bar{\boldsymbol{\mathbf{#1}}}}
\newcommand{\Bvectsf}[1]{\bar{\vectsf{#1}}}
\newcommand{\vectfrak}[1]{\vect{\mathfrak{#1}}}
\newcommand{\Tvect}[1]{\tilde{\boldsymbol{\mathbf{#1}}}}

 \newcommand{\boxend}{\hfill \ensuremath{\Box}}

\newcommand{\margin}[1]{\marginpar{\color{red}\tiny\ttfamily#1}}

\newtheorem{thm}{Theorem}[section]
\newtheorem{prop}{Proposition}[section]

\newtheorem{lem}{Lemma}[section]

\newcommand{\oprocendsymbol}{\hbox{$\bullet$}}
\newcommand{\oprocend}{\relax\ifmmode\else\unskip\hfill\fi\oprocendsymbol}

\newcommand{\solmaz}[1]{{\color{red}#1}}

\begin{document}

\maketitle

\begin{abstract}
When the underlying physical network layer in optimal network flow problems is a large graph, the associated optimization problem has a large set of decision variables. In this paper, we discuss how the cycle basis from graph theory can be used to reduce the size of this decision variable space. The idea is to eliminate the aggregated flow conservation constraint of these problems by explicitly characterizing its solutions in terms of the span of the columns of the transpose of a fundamental cycle basis matrix of the network plus a particular solution. We show that for any given input/output flow vector, a particular solution can be efficiently constructed from tracing any path that connects a source node to a sink node. We demonstrate our results over a minimum cost flow problem as well as an optimal power flow problem with storage and generation at the nodes. We also show that the new formulation of the minimum cost flow problem based on the cycle basis variables is amenable to a distributed solution. In this regard, we apply our method over a distributed alternating direction method of multipliers (ADMM) solution and demonstrate it over a numerical example.
\end{abstract}

\section{Introduction}
In a network flow problem, a physical system consisted of several routes between source and sink points transfers input flows from the source points to the sink points. The objective of optimal network flow problems mainly is to minimize the overall cost of transporting flow~\cite{bertsekas1998network}. Network flow problems appear in many important applications, such as communication networks~\cite{thomas2006minimum}, wireless sensor networks~\cite{wu2007robust}, wireless routing and resource allocation \cite{xiao2004simultaneous}, transportation systems~\cite{ba2015distributed} and power networks~\cite{nakayama2015distributed}. In power network problems, variants of  optimal network flow problems also include optimal generation and storage costs in their objectives~\cite{qin2016distributed,chandy2010simple,sun2015distributed,lavaei2012zero}.

With the advent of new technologies, the amount of available data and size of networks have been increasing. Such expansions in the size of physical networks result in increasing the size of optimization problems associated with optimal network flow problems. The number of decision variables has a direct relation with the time and space computation complexity of optimization problems. For large scale optimization problems, there has been efforts to use different variable reduction techniques to reduce problem size.
Variable fixing techniques~\cite{billionnet2004exact}, dominance technique~\cite{kellerer2004other} and constraint pairing techniques~\cite{osorio2002cutting} are some general reduction techniques in Integer Quadratic Problems (IQP). Moreover, in \cite{hua2008new} a new variable reduction techniques for IQP proposed which fixes some decision variables at zero without loosing optimality. 
In multi-objective optimization problems also it is shown that using data mining techniques it is possible to reduce less effective variables~\cite{esmaeili2010notice}. For evolutionary optimization problems,~\cite{wu2013complexity} presents how variable reduction techniques can be applied to obtain the variable relations from the partial derivatives of an optimization function. For optimization problems of the form~\eqref{eq::opt-affine}, eliminating affine equality constraint as discussed below is also a method to reduce the number of the search variables of the problem (c.f.~\cite{boyd2004convex})-- 
    \begin{align}\label{eq::opt-affine}
    &\vect{x}^\star = \underset{\vect{x}\in\real^n}{\argmin} ~\phi(\vect{x}),~~\text{s.t.}\quad  \vect{A}\vect{x} = \vect{b},~~\vect{g}(\vect{x}) \leq \vect{0},
    \end{align}    
where $\phi: \mathbb{R}^{n}\rightarrow \mathbb{R}$ and $g:\mathbb{R}^n\to\mathbb{R}^m$ are the cost function and the inequality constraint function, respectively, and $\vect{A} \in \mathbb{R}^{p \times n}$ satisfies $\rank(\vect{A}) = \rho\leq p < n$.
The affine feasible set for this optimization problem can be characterized as
\begin{equation}\label{eq::feas_case1}
    \{ \vect{x}\in\real^n \mid \vect{A} \vect{x} = \vect{b} \} = \{ \vect{F} \vect{z} + \vectsf{x}^{\text{p}} \mid \vect{z} \in \mathbb{R}^{n-\rho} \}.
\end{equation}
where $\vect{F} \in \mathbb{R}^{n \times (n-\rho)}$ is a matrix whose columns expand the null-space of $\vect{A}$ and $\vectsf{x}_{\text{p}} \in \mathbb{R}^{n}$ is a particular solution of $\vect{A}\vect{x} = \vect{b}$.
Then, $\vect{x}^\star$ in~\eqref{eq::opt-affine} satisfies  $\vect{x}^{\star} = \vect{F}\vect{z}^\star+\vectsf{x}_{\text{p}}$  where 
\begin{align}\label{eq::affine-eliminated-opt}
   \vect{z}^\star=&\underset{\vect{z}\in\real^{n-\rho}}{\argmin} \,\bar{\phi}(\vect{z}) = \phi(\vect{F}\vect{z}+\vectsf{x}^{\text{p}}),~~\text{s.t.}\\
 &  \Bvect{g}(\vect{z})=\vect{g}(\vect{F}\vect{z}+\vectsf{x}^{\text{p}})\leq \vect{0}.\nonumber
\end{align}
Compared to~\eqref{eq::opt-affine}, in~\eqref{eq::affine-eliminated-opt} not only  the equality constraint is eliminated but also the number of the search variables are reduced from $n$ to $n\!-\!\rho$. However, the lack of efficient methods to construct matrix~$\vect{F}$ and particular solution $\vect{x}_{\text{p}}$ can be an impediment in use of affine equality constraint elimination~method.

Optimal network flow problems are normally cast as a convex optimization problem where the cost is the sum of convex cost of flow through the arcs subject to capacity bounds for each arc and flow conservation equations at each node, resulting in optimization problems of the form~\eqref{eq::opt-affine}. In variations of the optimal network flow problem, the cost can be augmented to include the cost of e.g., generation and storage at nodes. The constraints can also be expanded to include other components of the problem. Nevertheless, in all network flow problems, an affine equality constraint that is always present is the flow conservation equation. To reduce the decision variables, one can use the aforementioned affine equality elimination approach, to eliminate the aggregated flow conservation equation from the network flow problems. In this paper, we discuss how  the cycle basis structure from graph theory (c.f.~\cite{dharwadker2011graph}) can be used to accomplish this elimination in an efficient manner. Minimum cycle bases have applications in many areas such as electrical
circuit theory \cite{chua1973optimally}, structural engineering \cite{cassell1976cycle}, surface reconstruction \cite{tewari2006meshing}.

In this paper, we show that all the solutions of the flow conservation equation is characterized explicitly in terms of the span of the columns of the transpose of the fundamental cycles basis matrix of the network plus a particular solution. The fundamental cycle basis of a graph can be computed in polynomial time using efficient algorithms such as those in~\cite{horton1987polynomial,amaldi2010efficient,borradaile2015min} (see Appendix for a breif review).
To compute a particular solution, we show that for any given input/output flow vector, a particular solution can be efficiently constructed from a set of elementary solutions each obtained from tracing a flow of value $1$ over the network from each node to a common particular sink node. We demonstrate our flow conservation equation elimination over a minimum cost flow problem as well as an optimal power flow problem with storage and generation at the nodes.

Parallel and distributed solutions are also sought as a method to solve large scale optimal network flow problems in an efficient manner. For example, a minimum cost network flow problem is solved in a distributed manner via dual sub-gradient descent in~\cite{bertsekas1998network}. For the same problem, a distributed second order method with a better convergence rate is proposed~in~\cite{zargham2014accelerated}. In~\cite{mota2015distributed}, a distributed algorithm based on the local domain ADMM approach is proposed for minimum cost flow problem. These algorithms are all arc-based, i.e., to solve the network flow problem in a distributed manner, each arc or group of arcs are assigned to  cyber-layer nodes. Then, the minimum cost network flow optimization problem is cast in a separable manner and solved by cyber-layer nodes in a cooperative way. Although in distributed algorithms the computational cost of the optimal flow problem is distributed among the cyber-layer nodes, the high number of decision variables normally translates to the high number of cyber nodes or large communication overhead between neighboring cyber nodes. Our next contribution in this paper is to show that the new formulation of the minimum cost network flow problem based on the cycle basis variables, which has a reduced set of search variables, is amenable to distributed solutions. Specifically, we demonstrate implantation of a distributed ADMM solution method (c.f. \cite{boyd2011distributed} and~\cite{mota2013communication}) over this new formulation. To implement this distributed solution we propose a cyber-layer whose nodes are defined based on the fundamental cycles of a cycle basis of the physical-layer graph.  A preliminary version of parts of our results in this paper has appeared in~\cite{RA-SSK-AR:16}.

\emph{Notations}: $\reals$, $\realpositive$, $\realnonnegative$,  and $\real_{\leq 0}$
denote the set of real, positive real, non-negative real, and non-positive real numbers, respectively. We let $\vect{A}^\top$ be the transpose of a matrix $\vect{A}$ and $[\vect{A}]_i$ indicate its $i^\text{th}$ column. We let $[\{z_k\}_{k=1}^n]$ be the column vector obtained from stacking the elements of an ordered set $\{z_k\}_{k=1}^n$. For network variables $\{p_i\}_{i=1}^N\subset \real$, defined over $N$ nodes, $N$ arcs, or $N$ cycles, we represent the aggregate vector of these variables by $\vect{p}=[\{p_i\}_{i=1}^N]\in\reals^N$.

%------------------------------------------------------

%------------------------------------------------------
\begin{figure}[t!]%[htbp]
  \centering   
  \subfloat[IEEE bus system 30]{
    \includegraphics[trim={10pt 0 10pt 0},clip,height=1.7in]{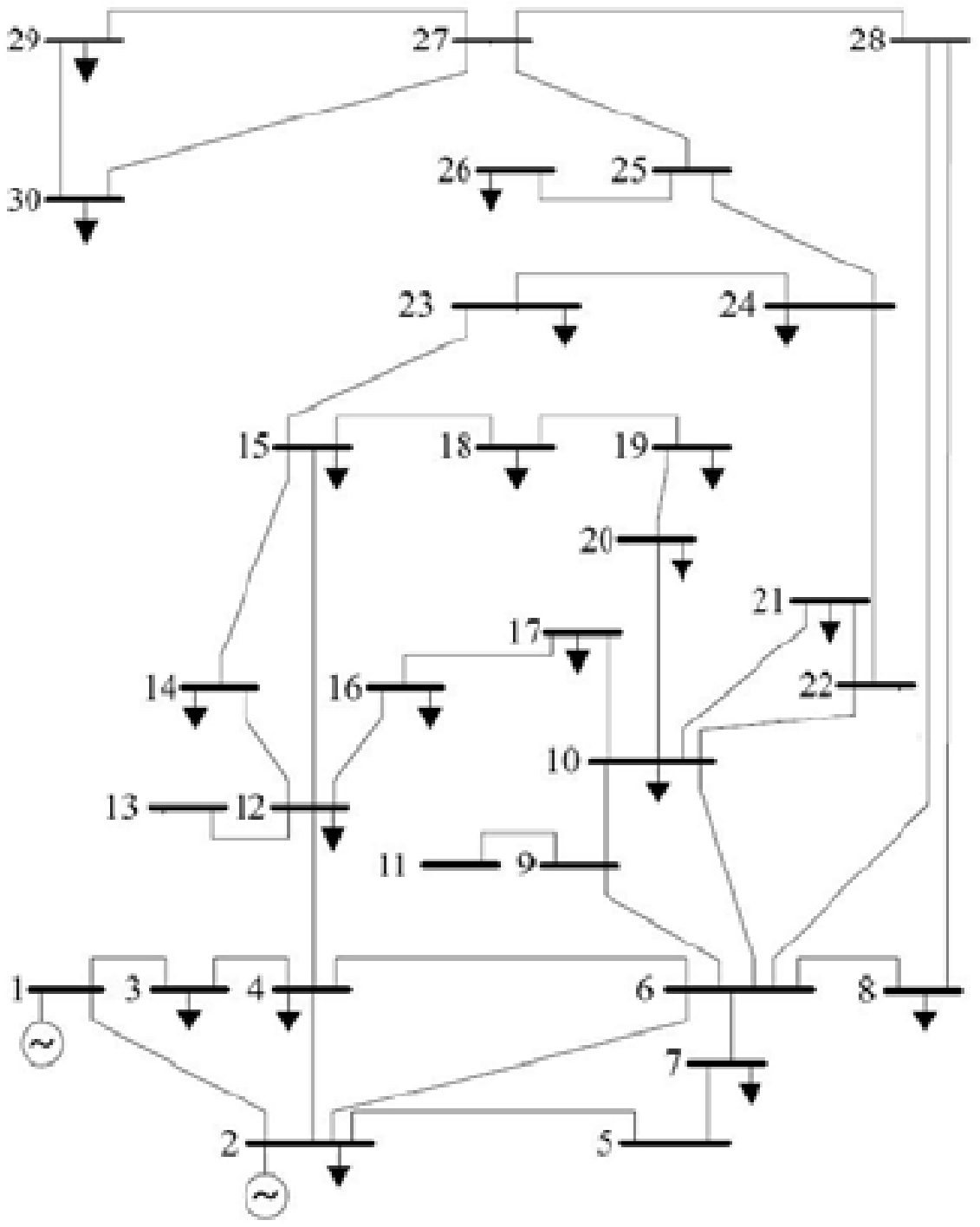}
  }\quad\quad\quad
\subfloat[Node-arc representation of IEEE bus system 30 with the cycles highlighted by dashed curves]{\begin{tikzpicture}[auto,thick,scale=0.55, every
    node/.style={scale=0.6}]
    \tikzstyle{mynode}=%
    [%
    minimum size=12pt,%
    inner sep=0pt,%
    outer sep=0pt,%
    draw,
    shape=circle%
    ]
    \draw        
    (-0.5,-4.5)  node[mynode,thin] (1) {{\scriptsize 1}}
    (0.5,-4.5)  node[mynode,thin] (2) {{\scriptsize 2}}
    (-0.5,-3.5) node[mynode,thin] (3) {{\scriptsize 3}}
    (0.5,-2)  node[mynode,thin] (4) {{\scriptsize 4}}
    (2,-2.5)  node[mynode,thin] (5) {{\scriptsize 5}}
    (3.5,-4.5)  node[mynode,thin] (6) {{\scriptsize 6}}
    (2,-3.8)  node[mynode,thin] (7) {{\scriptsize 7}}
    (6.5,-4.5)  node[mynode,thin] (8) {{\scriptsize 8}}
    (4.5,-2.5)  node[mynode,thin] (9) {{\scriptsize 9}}
    (3.5,-2)  node[mynode,thin] (10) {{\scriptsize 10}}
    (5.5,-3.5)  node[mynode,thin] (11) {{\scriptsize 11}}
    (0.5,-0.5)  node[mynode,thin] (12) {{\scriptsize 12}}
    (-0.5,-2)  node[mynode,thin] (13) {{\scriptsize 13}}
    (-0.5,-1) node[mynode,thin] (14) {{\scriptsize 14}}
    (-0.5,1)  node[mynode,thin] (15) {{\scriptsize 15}}
    (1.7,-1.0)  node[mynode,thin] (16) {{\scriptsize 16}}
    (2.5,-1.4)  node[mynode,thin] (17) {{\scriptsize 17}}
    (1.5,0.5)  node[mynode,thin] (18) {{\scriptsize 18}}
    (2.5,0)  node[mynode,thin] (19) {{\scriptsize 19}}
    (3.5,-0.5)  node[mynode,thin] (20) {{\scriptsize 20}}
    (3.5,0.5)  node[mynode,thin] (21) {{\scriptsize 21}}
    (4.5,-1)  node[mynode,thin] (22) {{\scriptsize 22}}
    (2.5,1.5)  node[mynode,thin] (23) {{\scriptsize 23}}
    (4.5,1)  node[mynode,thin] (24) {{\scriptsize 24}}
    (4.5,2)  node[mynode,thin] (25) {{\scriptsize 25}}
    (3,2)  node[mynode,thin] (26) {{\scriptsize 26}}
    (3,3)  node[mynode,thin] (27) {{\scriptsize 27}}
    (6.5,3)  node[mynode,thin] (28) {{\scriptsize 28}}
    (0.5,3)  node[mynode,thin] (29) {{\scriptsize 29}}
    (0.5,2)  node[mynode,thin] (30) {{\scriptsize 30}};
  \draw [gray,dashed] plot [smooth cycle] coordinates {($ (1) + (0.2,0.2) $)  ($(3)+(0.3,0.4)$)    ($(4) + (-0.2,-0.7)$) ($(2) + (-0.2,0.3)$)};
\draw [gray,dashed] plot [smooth cycle] coordinates {($ (2) + (0.1,0.5) $)  ($(4)+(0.2,-0.5)$)    ($(5) + (-0.4,-0.2)$)};
\draw [gray,dashed] plot [smooth cycle] coordinates {($ (2) + (0.5,0.2) $)  ($(5)+(-0.3,-0.7)$)    ($(7) + (-0.3,-0.2)$) ($(6) + (-0.6,0.2)$)};
\draw [gray,dashed] plot [smooth cycle] coordinates {($ (5) + (0.2,-0.5) $)  ($(6)+(-0.4,0.3)$)    ($(7) + (0.2,0.2)$)};
\draw [gray,dashed] plot [smooth cycle] coordinates {($ (6) + (0.2,0.6) $)  ($(9)+(-0.3,0)$)    ($(10) + (0.2,-0.4)$)};
\draw [gray,dashed] plot [smooth cycle] coordinates {($ (6) + (0.4,0.4) $)  ($(9)+(0,-0.3)$)    ($(11) + (-0.4,0)$)};
\draw [gray,dashed] plot [smooth cycle] coordinates {($ (12) + (-0.3,0) $)  ($(14)+(0.3,0.4)$)    ($(15) + (0.1,-0.4)$)};
\draw [gray,dashed] plot [smooth cycle] coordinates {($ (27) + (-0.7,-0.1) $)  ($(29)+(0.3,-0.2)$)    ($(30) + (0.3,0.3)$)};
\draw [gray,dashed] plot [smooth cycle] coordinates {($ (4) + (0.3,0.3) $)  ($(5)+(0,0.4)$)    ($(6) + (-0.3,0.7)$) ($(10) + (-0.4,-0.3)$) ($(17) + (0,-0.3)$) ($(16) + (0,-0.3)$) ($(12) + (0.3,-0.3)$)};
\draw [gray,dashed] plot [smooth cycle] coordinates {($ (12) + (0,0.3) $)  ($(16)+(0,0.3)$)    ($(17) + (0,0.3)$) ($(10) + (-0.3,0.3)$) ($(20) + (-0.3,-0.2)$) ($(19) + (0,-0.2)$) ($(18) + (0,-0.2)$) ($(15) + (0.4,-0.3)$)};
\draw [gray,dashed] plot [smooth cycle] coordinates {($ (15) + (0.6,0) $)  ($(18)+(0,0.2)$)    ($(19) + (0,0.2)$) ($(20) + (0.2,0.2)$) ($(10) + (0.2,0.4)$) ($(22) + (-0.4,0)$) ($(24) + (-0.3,-0.3)$) ($(23) + (0,-0.2)$)};
\draw [gray,dashed] plot [smooth cycle] coordinates {($ (6) + (0.9,0.2) $)  ($(11)+(0.4,0)$)    ($(9) + (0.3,0.2)$) ($(10) + (0.6,0)$) ($(22) + (0.3,0)$) ($(24) + (0.3,0)$) ($(25) + (0.3,0.2)$) ($(27) + (0.7,-0.3)$) ($(28) + (-0.7,-0.7)$) ($(8) + (-0.3,0.8)$)};

\draw[thick, ->] (1) -- (2) node [midway, above, sloped] (e1) {$ $};
\draw[thick, ->] (1) -- (3) node [midway, below, sloped] (e2) {$ $};
\draw[thick, ->] (2) -- (4) node [midway, below, sloped] (e3) {$ $};
\draw[thick, ->] (2) -- (5) node [midway, below, sloped] (e4) {$ $};
\draw[thick, ->] (2) -- (6) node [midway, below, sloped] (e5) {$ $};
\draw[thick, ->] (3) -- (4) node [midway, below, sloped] (e6) {$ $};
\draw[thick, ->] (4) -- (5) node [midway, below, sloped] (e7) {$ $};
\draw[thick, ->] (4) -- (12) node [midway, below, sloped] (e8) {$ $};
\draw[thick, ->] (4) -- (13) node [midway, below, sloped] (e9) {$ $};
\draw[thick, ->] (5) -- (6) node [midway, below, sloped] (e10) {$ $};
\draw[thick, ->] (5) -- (7) node [midway, below, sloped] (e41) {$ $};
\draw[thick, ->] (6) -- (7) node [midway, below, sloped] (e11) {$ $};
\draw[thick, ->] (6) -- (8) node [midway, below, sloped] (e12) {$ $};
\draw[thick, ->] (6) -- (9) node [midway, below, sloped] (e13) {$ $};
\draw[thick, ->] (6) -- (10) node [midway, below, sloped] (e14) {$ $};
\draw[thick, ->] (6) -- (11) node [midway, below, sloped] (e15) {$ $};
\draw[thick, ->] (8) -- (28) node [midway, below, sloped] (e16) {$ $};
\draw[thick, ->] (9) -- (10) node [midway, below, sloped] (e17) {$ $};
\draw[thick, ->] (9) -- (11) node [midway, above, sloped] (e18) {$ $};
\draw[thick, ->] (10) -- (17) node [midway, above, sloped] (e19) {$ $};
\draw[thick, ->] (10) -- (20) node [midway, above, sloped] (e20) {$ $};
\draw[thick, ->] (10) -- (22) node [midway, above, sloped] (e21) {$ $};
\draw[thick, ->] (12) -- (14) node [midway, above, sloped] (e22) {$ $};
\draw[thick, ->] (12) -- (15) node [midway, above, sloped] (e23) {$ $};
\draw[thick, ->] (12) -- (16) node [midway, above, sloped] (e24) {$ $};
\draw[thick, ->] (14) -- (15) node [midway, above, sloped] (e25) {$ $};
\draw[thick, ->] (15) -- (18) node [midway, above, sloped] (e26) {$ $};
\draw[thick, ->] (15) -- (23) node [midway, above, sloped] (e27) {$ $};
\draw[thick, ->] (16) -- (17) node [midway, above, sloped] (e28) {$ $};
\draw[thick, ->] (18) -- (19) node [midway, above, sloped] (e29) {$ $};
\draw[thick, ->] (19) -- (20) node [midway, above, sloped] (e30) {$ $};
\draw[thick, ->] (20) -- (21) node [midway, above, sloped] (e31) {$ $};
\draw[thick, ->] (22) -- (24) node [midway, above, sloped] (e32) {$ $};
\draw[thick, ->] (23) -- (24) node [midway, above, sloped] (e33) {$ $};
\draw[thick, ->] (24) -- (25) node [midway, above, sloped] (e34) {$ $};
\draw[thick, ->] (25) -- (26) node [midway, above, sloped] (e35) {$ $};
\draw[thick, ->] (25) -- (27) node [midway, above, sloped] (e36) {$ $};
\draw[thick, ->] (27) -- (28) node [midway, above, sloped] (e37) {$ $};
\draw[thick, ->] (27) -- (29) node [midway, above, sloped] (e38) {$ $};
\draw[thick, ->] (27) -- (30) node [midway, above, sloped] (e39) {$ $};
\draw[thick, ->] (29) -- (30) node [midway, below, sloped] (e40) {$ $};

\tikzstyle{vecArrow} = [thick, decoration={markings,mark=at position
   1 with {\arrow[semithick]{open triangle 60}}},
   double distance=1.4pt, shorten >= 5.5pt,
   preaction = {decorate},
   postaction = {draw,line width=1.4pt, white,shorten >= 4.5pt}];
\tikzstyle{innerWhite} = [semithick, white,line width=1.4pt, shorten >= 4.5pt];
    \end{tikzpicture} }  
\caption{The graph related to IEEE bus system 30 with 41 arcs, 30 nodes and 12 cycles.}\label{fig::min-cycle-example}
\end{figure}
\section{A review of cycle basis in graphs}\label{sec::graph_Theory}
In this section, following~\cite{dharwadker2011graph}, we review our graph related terminology and conventions. We also introduce our graph related notations. We represent a graph of  $n$ nodes and $m$ arcs with $\GG=(\mathcal{V},\mathcal{E})$, where $\mathcal{V} = \{ v_1,v_2, \cdots , v_n\}$ is the node set and  $\mathcal{E}=\{e_1,\cdots,e_m\}\in\VV\times\VV$ is the arc set. The graph is assumed to be undirected and with no self-loop.  A \emph{walk} is an alternating sequence of nodes and connecting arcs. A \emph{path} is a walk that does not include any node twice, except for its first and last nodes which can be the same. A graph is \emph{connected} if there is a path from its every node to every other node.  The degree of a node in a graph is the total number of arcs connected to that node. When there is an orientation assigned to the arcs of a graph $\GG=(\VV,\mathcal{E})$, we represent the \emph{oriented graph} by $\GG^\text{o}=(\mathcal{V},\mathcal{E}^{\text{o}})$. We write $e_k=(v_i,v_j)\in\mathcal{E}^{\text{o}}$ if arc $e_k$ points from node $v_i$ towards node $v_j$. If $(v_i,v_j)\in\mathcal{E}^{\text{o}}$ then $(v_j,v_i)\notin\mathcal{E}^{\text{o}}$, i.e., there is no symmetric arc in the oriented graph. For $\GG^{\text{o}}$, the \emph{oriented incidence} matrix is the matrix $\vectfrak{I}^\text{o} \in \mathbb{R}^{|\mathcal{V}| \times |\mathcal{E}|}$, where $\vectfrak{I}^\text{o}_{ij}=1$ if arc $e_j$ leaves node $v_i$, $\vectfrak{I}^\text{o}_{ij}=-1$ if arc $e_j$ enters node $v_i$, otherwise $\vectfrak{I}^\text{o}_{ij}=0$.
For a connected graph of $n$ nodes with a given orientation, the rank of  $\vectfrak{I}^\text{o}$ is $n-1$.

A \emph{cycle} of $\mathcal{G}$ is any sub-graph in which each node has even degree.  A \emph{simple cycle} is a path that begins and ends on the same node with no other repetitions of nodes. A cycle vector $\vectsf{c}\in\real^m$ is a binary vector with $\vectsf{c}_{i} =1$ if $e_i$ is in the cycle and $\vectsf{c}_{i} =0$, otherwise.  
A \emph{cycle basis} of $\GG$ is a set of simple cycles that forms a basis of the cycle space of $\GG$. Every cycle in a given cycle basis is called a \emph{fundamental cycle}.  A fundamental cycle basis of a graph is constructed by its spanning tree, in a way that cycles formed by a combination of a path in the tree and a single arc outside of the tree. For every arc outside of the tree, there exist one cycle. Each cycle generated in this way is independent of other cycles, because it has one arc, not exist in other cycles (see Fig.~\ref{fig::min-cycle-example}). The dimension of cycle basis of a graph is $\mu=m-n+1$. For cycles in an oriented graph  $\GG^\text{o}=(\VV,\EE^\text{o})$, we assign the counter clockwise direction as positive cycles orientation and define the \emph{oriented cycle vector} $\vectsf{c}^\text{o}\in\real^m$ with $\vectsf{c}^\text{o}_{i} =1$ if $e_i$ is in the cycle and aligned with its direction, $\vectsf{c}^\text{o}_{i} =-1$ if $e_i$ is in the cycle but opposing the direction of the cycle and finally $\vectsf{c}^\text{o}_{i} =0$ if $e_i$ is not in the cycle. Given a cycle basis, we define the \emph{oriented fundamental cycle basis matrix} $\vectsf{B}^\text{of}\in \real^{\mu \times m}$ as a matrix whose rows are each the transpose of the oriented cycle vector of the fundamental cycles of this cycle basis. This matrix satisfies $\rank(\vectsf{B}^\text{of})=\mu$.
\begin{thm}[relationship between the oriented incidence matrix and an oriented cycle vector ~(c.f.~\cite{dharwadker2011graph})]\label{thm::I-B}
In an oriented graph $\GG^\text{o}$, every oriented cycle vector $\vectsf{c}^o$ is orthogonal to every row of oriented incident matrix $ \vectfrak{I}^{\text{o}}$, i.e., $ \vectfrak{I}^{\text{o}}\,\vectsf{c}^o=\vect{0}_n$. 
\boxend
\end{thm}

\section{Decision variable reduction in network flow problems}\label{sec::prob_def}
In this section, we show how two well-known network flow problems can benefit from affine equality elimination method to reduce their search variables. We study our optimal network flow problems of interest over a network of $n$ nodes where each node is connected to a subset of other nodes through some form of routes. For example, in a power network the route is a transmission line, while in a transportation network the route is the road connecting two conjunction nodes on the road map. The physical layer topology is described by a connected graph $\mathcal{G}_{\text{physic}}=(\VV_{\text{physic}},\EE_{\text{physic}})$, where $|\VV_{\text{physic}}|=n$ and $|\EE_{\text{physic}}|=m$. The flow can travel in both directions in every route, however, we assume a pre-specified positive orientation for each route and based on it we describe the flow network in the physical layer by the oriented version of $\mathcal{G}_{\text{physic}}$, i.e., $\mathcal{G}^\text{o}_{\text{physic}}=(\VV_{\text{physic}},\EE^\text{o}_{\text{physic}})$.
This physical network transfers flow(s) from a set of source nodes to a set of sink nodes (see physical layers in Fig.~\ref{fig::CPS} and Fig.~\ref{fig::min-cycle-numeric-example}), while respecting the conservation of the flow constraints, i.e.,  the total inflow into each node must be equal to the total outflow from that node. We let $x_i$ be the flow across the arc $e_i\in\EE^\text{o}_{\text{physic}}=\{e_1,e_2,\cdots,e_m\}$.  Every arc $e_i\in\EE^\text{o}_{\text{physic}}$  has a pre-specified capacity, i.e.,   $\mathsf{b}_{i}\leq x_{i}\leq \mathsf{c}_{i}$, for some known $\mathsf{b}_{i},\mathsf{c}_{i}\in\real$.
For any external flow $\mathsf{f}_i$, we use the sign convention of $\mathsf{f}_i>0$ for input flow and $\mathsf{f}_i<0$ for output flow, and $\mathsf{f}_i=0$ otherwise.

\subsection{Minimum cost flow problem}
We consider a minimum cost flow problem over $\mathcal{G}^\text{o}_{\text{physic}}$ with a given set of input and output flows at specific source and sink points. In this problem, there is a convex cost $\phi_{i}:\real\to\real$ associated with flow across each arc $e_i\in\mathcal{E}^{o}_{\text{physic}}$, and our objective is  to find the network minimizer $\vect{x}^\star\in\real^m$ in the following optimization problem
 \begin{subequations}\label{eq::prob_def1}
\begin{align}
   \vect{x}^\star&=\underset{\vect{x}\in\real^m}{\argmin}~ \phi(\vect{x})=\sum\nolimits_{i=1}^m\phi_{i}(x_{i}),~~\text{s.t.,}\label{eq::prob_def-cost}\\
&\quad \sum\nolimits_{j=1}^m \vectfrak{I}^\text{o}_{ij}\,x_{j}=\mathsf{f}_i,\quad i\until{n},\label{eq::prob_def-flowconserv}\\
  &\quad  \mathsf{b}_{j}\leq x_{j}\leq \mathsf{c}_{j},~\quad\quad\quad j\in\{1,\cdots,m\},\label{eq::prob_def-capacity}
  \end{align}
\end{subequations}
where $\vectsf{f}=(\mathsf{f}_1,\cdots,\mathsf{f}_n)^\top$ is the given input/output flow vector which satisfies $\sum_{i=1}^n \mathsf{f}_i= 0$.  Here,~\eqref{eq::prob_def-flowconserv} captures the flow conservation at nodes across the network and~\eqref{eq::prob_def-capacity} describes the arc capacity constraints. 
\begin{comment}
\margin{we may not need this}\solmaz{For a given network and capacity bounds, maximum (resp. minimum) network flow problem gives an upper and (resp. lower) bound on the admissible ranges of input flow $f_{\text{in}}$ such that the feasible set~\eqref{eq::feasible-set} is always non-empty. Maximum (also minimum) flow of a network can be find using Edmonds-Krap algorithm in $\mathcal{O}(nm^{2})$ in a central way~\cite{edmonds1972theoretical}.}
\end{comment}
The number of search variables in the optimization problem~\eqref{eq::prob_def1} is equal to the number of the arcs of the network, i.e., $|\mathcal{E}^{o}_{\text{physic}}|=m$. 
Our result below uses Theorem~\ref{thm::I-B} to  eliminate the affine equality constrains~\eqref{eq::prob_def-flowconserv} and reduce the search variables to $m-n+1$.

\begin{thm}[Eliminating the flow conservation constraint from~\eqref{eq::prob_def1}]\label{thm::main-flow-eliminat}
Consider the optimal network flow problem~\eqref{eq::prob_def1} over a connected physical network  $\GG_{\text{\rm{physic}}}^\text{o}$. Then, $\vect{x}^\star$ in~\eqref{eq::opt-affine} satisfies $\vect{x}^{\star} = \vectsf{B}^\text{\rm{of}}\,^\top\vect{z}^\star+\vectsf{x}^{\text{\rm{p}}}$  where 
\begin{align}\label{eq::affine-eliminated-prob_def}
   \vect{z}^\star=&\underset{\vect{z}\in\real^{m-n+1}}{\argmin} \phi(\vect{z}) = \sum_{i=1}^m\phi_i(\vect{z}^\top[\vectsf{B}^\text{\rm{of}}]_i+\mathsf{x}_i^{\text{p}}),~~\text{s.t.}\\
 &  \quad  \mathsf{b}_{j}\leq \vect{z}^\top[\vectsf{B}^\text{\rm{of}}]_j+\mathsf{x}_j^{\text{p}}\leq \mathsf{c}_{j},~\quad j\in\{1,\cdots,m\},\nonumber
\end{align}
and $\vectsf{B}^\text{\rm{of}}$ is an oriented  fundamental cycle matrix of $\mathcal{G}^\text{o}_{\text{physic}}$ and $\vectsf{x}^{\text{\rm{p}}}$ is a particular solution of $\vectfrak{I}^{\text{o}}\,\vect{x}=\vectsf{f}$.
\end{thm}
\begin{proof}The equality constraint~\eqref{eq::prob_def-flowconserv} in aggregated form is  
\begin{align}\label{eq::prob_def-flowconserv-compact}
    \vectfrak{I}^{\text{o}}\,\vect{x}=\vectsf{f}.
\end{align}
Invoking the same argument that is used to relate solutions of the optimization problem~\eqref{eq::opt-affine} to those of~\eqref{eq::affine-eliminated-opt}, the proof relays on showing that the null-space of $\vectfrak{I}^{\text{o}}$ is spanned by columns of $\vectsf{B}^\text{\rm{of}}\,^\top$. By virtue of Theorem~\ref{thm::I-B}, we have $\vectfrak{I}^{\text{o}}\vectsf{B}^\text{\rm{of}}\,^\top=\vect{0}$. Recall that for a connected oriented graph $\rank(\vectfrak{I}^{\text{o}})=n-1$. Because  $\rank(\vectsf{B}^\text{\rm{of}}\,^\top)=m-n+1$, null-space of $\vectfrak{I}^{\text{o}}\in\real^{n\times m}$ is spanned by columns of $\vectsf{B}^\text{\rm{of}}\,^\top$.
This completes our proof.
\end{proof}

The effectiveness of the decision variable reduction method in Theorem~\ref{thm::main-flow-eliminat} depends on how efficiently one can  construct matrix $\vectsf{B}^\text{\rm{of}}$ and particular solution $\vectsf{x}^{\text{\rm{p}}}$, especially in large scale networks. In regards to matrix $\vectsf{B}^\text{\rm{of}}$, as reviewed in Appendix, there are efficient  algorithms that can construct cycle basis in polynomial time. Next, we propose a simple method to construct a particular solution $\vectsf{x}^{\text{\rm{p}}}$ using graph topology. Our method relies on the  superposition property of linear algebra equations, and the fact that a particular solution for a unit flow $\mathsf{f}_i=1$ entering the network at node $v_i$ and leaving it at node $v_j$ can simply be constructed by assuming that $\mathsf{f}_i=1$ flows along a path from node $v_i$ to node $v_j$.

\begin{lem}[Particular solution of~\eqref{eq::prob_def-flowconserv-compact}]\label{lem::particular_sol}
Given an input/output flow vector $\vectsf{f}$ over $\GG_{\text{physic}}^{\text{o}}$ which satisfies $\sum_{i=1}^n \mathsf{f}_i= 0$, a particular solution for~\eqref{eq::prob_def-flowconserv-compact} is $\vectsf{x}^{\text{p}}=\sum_{i=1}^{n-1}\mathsf{f}_i\, \Bvectsf{x}^{\text{p},v_i}$. Here, $\Bvectsf{x}^{\text{p},v_i}\in\real^m$, $i\until{n-1}$, is constructed from a path that connects node $v_i$ to node $v_n$ such that
$\mathsf{x}^{\text{p},v_i}_j=1$ (resp. $\mathsf{x}_{j}^{\text{p},v_i}=-1$) if $e_j$, $j\until{m}$, is on this path and is along (resp. opposing) the direction  of the path, otherwise $\mathsf{x}_{j}^{\text{p},v_i}=0$.
\end{lem}
\begin{proof}
For every $i\until{n-1}$, consider a virtual scenario where a unit flow $\bar{\mathsf{f}}_i=1$ enters the network at node $v_i$ and leaves it at node $v_n$. Using a simple flow tracing over the network we can see that $\Bvectsf{x}^{\text{p},v_i}\in\real^m$ as described in the statement satisfies $\vectfrak{I}^{\text{o}}\,\Bvectsf{x}^{\text{p},v_i}=\Bvectsf{f}^{v_i}$, $i\until{n-1}$
where $\bar{\mathsf{f}}^{v_i}_i=1$, $\bar{\mathsf{f}}^{v_i}_n=-1$ and $\bar{\mathsf{f}}^{v_i}_j=0$, $j\until{n}\backslash\{i,n\}$.
For a given network flow vector  $\vectsf{f}$ because $\mathsf{f}_n=-\sum_{i=1}^{n-1} \mathsf{f}_i$, we can write $\vectsf{f}=\sum_{i=1}^{n-1} \mathsf{f}_i \,\Bvectsf{f}^{v_i}$. Therefore, 
$\vectfrak{I}^{\text{o}} \sum\nolimits_{i=1}^{n-1}\mathsf{f}_i\, \Bvectsf{x}^{\text{p}}_{v_i}=\sum\nolimits_{i=1}^{n-1} \mathsf{f}_i \,\Bvectsf{f}^{v_i}=\vectsf{f},$
which completes our proof.
\end{proof}
A few remarks are in order regarding the particular solution. First, note that construction of the `elementary' particular solution set $\{\Bvectsf{x}^{\text{p},v_i}\}_{i=1}^{n-1}$ is regardless of the value of the network flow vector $\vect{f}$.  Second, for problems with single source and single sink nodes, we can label the source node $v_1$ and the sink node $v_n$ and compute the elementary solution set only for node $v_1$. More particularly, if in a given network flow problem the sink and the source nodes are fixed we only need to compute the elementary particulars solution set for the collection of sink and source nodes (see Section~\ref{sec::numeric} for an illustrative numerical example). Finally, to obtain sparse elementary solutions, we can use a shortest path between nodes $v_i$ and $v_n$.

\begin{figure}[t!]
  \centering   
    \includegraphics[height=1.4in]{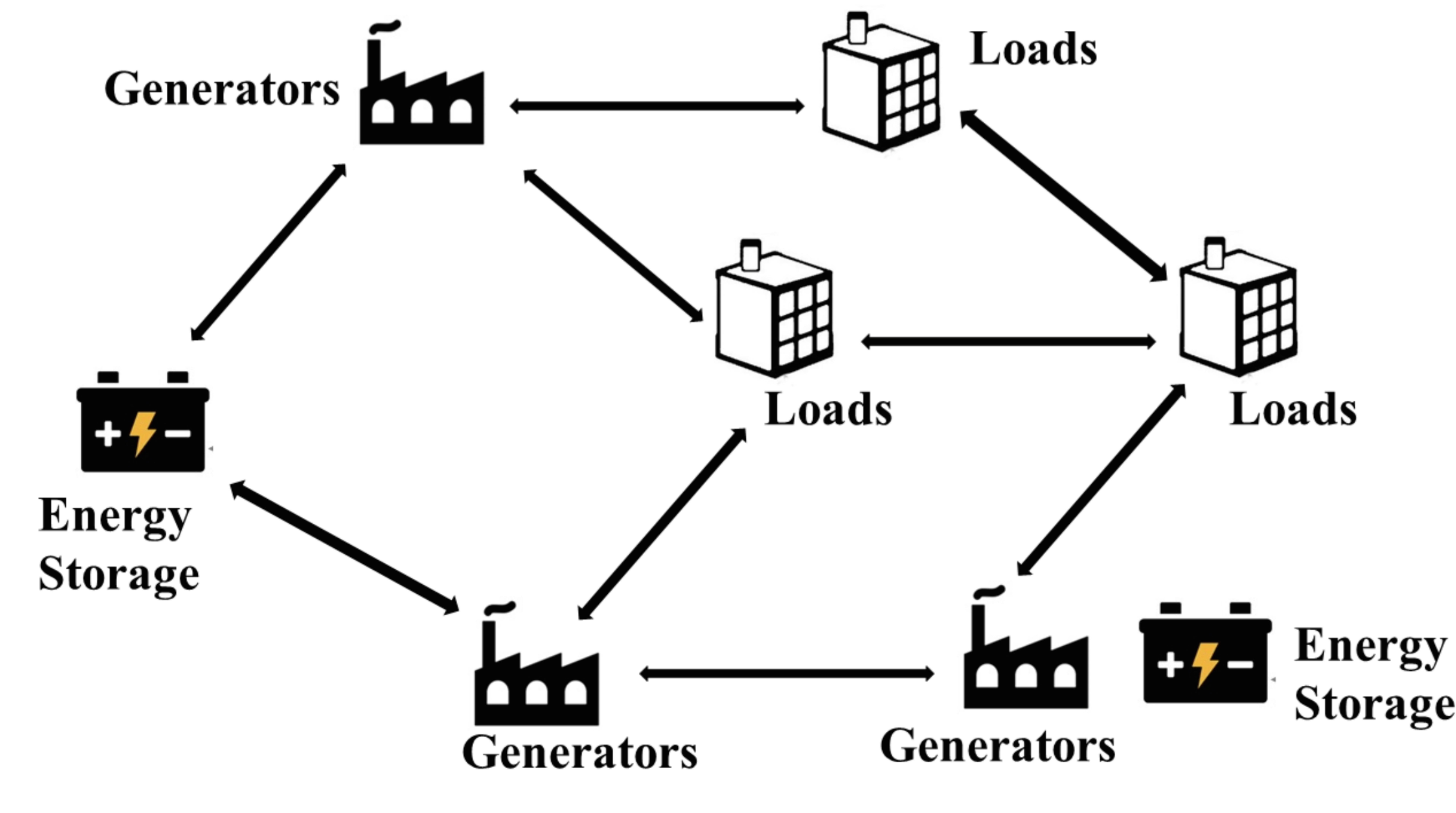}
    \caption{An schematic representation of a network flow problem with generation and storage at nodes}
\label{fig::opt-pow-stor-gen}\end{figure}
\subsection{Optimal power flow with storage and generation at nodes}
Next, we consider an optimal power flow problem over a network described by $\mathcal{G}^\text{o}_{\text{physic}}$ with storage, generation and load at its nodes (see Fig.~\ref{fig::opt-pow-stor-gen}). The objective in this problem is to minimize the cost of power generation along with energy loss at the transmission lines  over some finite time interval $\mathcal{T}=\{1,\cdots,T\}$.  Mathematical modeling of this problem over various scenarios including deterministic and stochastic generators is considered in the literature, e.g.,~\cite{qin2016distributed,chandy2010simple,sun2015distributed,lavaei2012zero}. 
All these models, at each time $t\in\mathcal{T}$, include a flow conservation equation at each node. As a result for a network with $m$ arcs, the flow conservation equation introduces $m\,|\mathcal{T}|$ decision variables into the optimization problem.

In our study below, without loss of generality, we use the deterministic form (no renewable generation source) of the optimal network flow problem studied in~\cite{qin2016distributed}, which states the problem as a direct current (DC) power flow problem (see~\eqref{eq::opt-pow-flow}). Without loss of generality, we assume that at each time $t\in\mathcal{T}$, each node $v_i\in\VV_{\text{physic}}$ has a generator which  supplies a bounded $\delta_{i}(t)$ power, a battery with a bounded storage level $s_i(t)$ and a charge/discharge variable $u_i(t)$, and a known demand $\mathsf{d}_i(t)\in\real_{\leq0}$.
 If a node does not have either of the generation, storage, or load components, we simply remove the respective variables from our formulation below.
Then, given a known load profile $\{\vectsf{d}(t)\}_{t=1}^T$, where $\vectsf{d}(t)=[\{\mathsf{d}_i(t)\}_{i=1}^n$], the optimization problem of interest is
\begin{subequations}\label{eq::opt-pow-flow}
    \begin{align}
    &\{\vect{x}^{\star}(t),\vect{\delta}^{\star}(t),\vect{u}^{\star}(t),\vect{s}^\star(t),\vect{\theta}^{\star}(t)\}_{t=1}^T=\\
       &\argmin~\frac{1}{T}\sum\nolimits_{t=1}^T\!\Big(\sum\nolimits_{j=1}^n\! g_j(\delta_j(t))+\!\sum\nolimits_{i=1}^m \!\phi_i(x_i(t))\Big),~\text{s.t.}\nonumber\\
      & \text{for~}t\in\mathcal{T},~~i\until{n}\nonumber\\ 
        & ~~\sum\nolimits_{j=1}^m \vectfrak{I}^\text{o}_{ij}\,x_i(t)=\delta_i(t) + u_i(t)+\mathsf{d}_i(t),\label{eq::flow_case2}\\
        &~~s_{i}(t+1) = \lambda_i s_i(t) + u_i(t),\\
        &~~\vect{B}_{ij}(\theta_i(t) - \theta_j(t)) = x_{i}(t),\quad j \in \mathcal{N}^{e}(i),\\
        &~~\underline{\mathsf{\delta}}_i \leq \!\delta_i(t)\! \leq \bar{\mathsf{\delta}}_i,~~\underline{\mathsf{u}}_i \leq \!u_i(t)\! \leq \bar{\mathsf{u}}_i,~~\underline{\mathsf{s}}_i \leq \!s_i(t)\! \leq \bar{\mathsf{s}}_i,\\       
        &~~ \mathsf{b}_{j}\leq x_{j}(t)\leq \mathsf{c}_{j},~~\quad\quad j\in\{1,\cdots,m\},
    \end{align}
\end{subequations}
where, $g_j$, $j\until{n}$ and $\phi_i$, $i\until{m}$ are cost functions for generators and power flows,  respectively. Here, $\mathcal{N}^{e}(i)$ is the set of the nodes that are connected to node $v_i$ through an arc, and $\lambda_i \in (0,1]$ is the storage energy dissipation factor. Moreover, $(\underline{\mathsf{s}}_i,\bar{\mathsf{s}}_i)\in\real\times\real$, $(\underline{\mathsf{u}}_i,\bar{\mathsf{u}}_i)\in\real\times\real$, $(\underline{\mathsf{\delta}}_i,\bar{\mathsf{\delta}}_i)\in\real\times\real$, are, respectively, known (lower bound, upper bound) values on storage level, battery charge/discharge and power generation by the generator. Finally $\vect{B}_{ij}(\theta_i(t) - \theta_j(t))$ is the DC approximation for alternating current power flow. Here, $\theta_i(t)$ is the voltage phase angle of node (bus) $v_i\in\VV_{\text{physic}}$ at time $t$ and $\vect{B} \in \reals^{n \times n}$ is the imaginary part of the admittance matrix under DC assumption.

The following result shows that, the number of search variables related to the flow across the arcs in the optimization problem~\eqref{eq::opt-pow-flow} can be reduced from $m|\mathcal{T}|$ to $(m-n+1)\mathcal{T}$ via eliminating the flow conservation constraint at each $t\in\mathcal{T}$. An interesting, observation in the result below is that in order to eliminate the flow conservation equations~\eqref{eq::flow_case2}, we need to introduce a new set of affine equality constraint~\eqref{eq::new-sum} which ensure balance between the external input and output flows.

\begin{prop}[Eliminating the flow conservation constraint from~\eqref{eq::opt-pow-flow}]\label{thm::main-flow-eliminat-case2}
Consider the optimal power flow problem~\eqref{eq::opt-pow-flow} over a physical network described by $\GG_{\text{\rm{physic}}}^\text{o}$ with a given set of loads $\{\vectsf{d}(t)\}_{t=1}^T$. Then, $\{\vectsf{B}^\text{\rm{of}}\,^\top\vect{z}^\star(t)+\vectsf{x}^{\text{\rm{p}}}(\vect{\delta}^{\star}(t),\vect{u}^{\star}(t),\vectsf{d}(t)),\vect{u}^{\star}(t),\vect{s}^\star(t),\vect{\theta}^{\star}(t)\}_{t=1}^T$  is a minimizer of the optimization problem~\eqref{eq::opt-pow-flow} where
\begin{subequations}\label{eq::opt-pow-flow-equiv}
    \begin{align}
    &\{\vect{z}^{\star}(t),\vect{\delta}^{\star}(t),\vect{u}^{\star}(t),\vect{s}^\star(t),\vect{\theta}^{\star}(t)\}_{t=1}^T=\nonumber\\
       &\quad \argmin~\frac{1}{T}\sum\nolimits_{t=1}^T\!\Big(\sum\nolimits_{j=1}^n\! g_j(\delta_j(t))\,+\\
       &~~ ~\quad\quad\quad\quad\sum\nolimits_{i=1}^m \!\phi_i(
       \vect{z}(t)^\top[\vectsf{B}^\text{\rm{of}}]_i+\mathsf{x}_i^{\text{p}}(\vect{\delta}(t),\vect{u}(t),\vectsf{d}(t))\Big),~\text{s.t.}\nonumber\\
      & \text{for~}t\in\mathcal{T},~~i\until{n}\nonumber\\ 
      &~\sum\nolimits_{j=1}^n (\delta_j(t) + u_j(t)+\mathsf{d}_j(t))=0,\label{eq::new-sum}\\
        &~s_{i}(t+1) = \lambda_i s_i(t) + u_i(t),\\
        &~\vect{B}_{ij}(\theta_i(t) - \theta_j(t)) = \vect{z}(t)^\top[\vectsf{B}^\text{\rm{of}}]_i+\mathsf{x}_i^{\text{p}}(\vect{\delta}(t),\vect{u}(t),\vectsf{d}(t)),\nonumber\\
        &\quad\quad\quad\quad j \in \mathcal{N}^{e}(i),\\
        &~\underline{\mathsf{\delta}}_i \leq \!\delta_i(t)\! \leq \bar{\mathsf{\delta}}_i,~~\underline{\mathsf{u}}_i \leq \!u_i(t)\! \leq \bar{\mathsf{u}}_i,~~\underline{\mathsf{s}}_i \leq \!s_i(t)\! \leq \bar{\mathsf{s}}_i,\\       
        &~ \mathsf{b}_{j}\leq \vect{z}(t)^\top[\vectsf{B}^\text{\rm{of}}]_j+\mathsf{x}_j^{\text{p}}(\vect{\delta}(t),\vect{u}(t),\vectsf{d}(t))\leq \vectsf{c}_{j},\\
        &\quad\quad\quad\quad\quad\quad\quad\quad\quad\quad\quad\quad\quad\quad\quad~ j\in\{1,\cdots,m\}.\nonumber
    \end{align}
\end{subequations}
Here, $\vectsf{B}^\text{\rm{of}}$ is a fundamental cycle basis matrix of $\mathcal{G}^\text{o}_{\text{physic}}$ and $\vectsf{x}^{\text{\rm{p}}}(\vect{\delta}(t),\vect{u}(t),\vectsf{d}(t))=\sum_{i=1}^{n-1}(\delta_i(t) + u_i(t)+\mathsf{d}_i(t))\, \Bvectsf{x}^{\text{p},v_i}$, where $\{\Bvectsf{x}^{\text{p},v_i}\}_{i=1}^{n-1}$ is as described in Lemma~\ref{lem::particular_sol}.
\end{prop}
\begin{proof}
The equality constraint~\eqref{eq::flow_case2} in aggregated form is 
\begin{align}\label{eq::flow-aggregate-case2}
   \vectfrak{I}^{\text{o}}\,\vect{x}(t)=\vect{\delta}(t)+\vect{u}(t)+\vectsf{d}(t),\quad\quad t\in\mathcal{T}.
\end{align}
 Note that rank of $\vectfrak{I}^{\text{o}}\in\real^{n\times m}$ is $n-1$ and also that $\vect{1}_n^\top\vectfrak{I}^{\text{o}}=\vect{0}$, where $\vect{1}_n$ is the vector of $n$ ones. Left multiplying~\eqref{eq::flow-aggregate-case2} by $\vect{1}_n^\top$ results in~\eqref{eq::new-sum}. Then, for $\vect{\delta}(t),\vect{u}(t),\vectsf{d}(t)\in\real^n$ that satisfy ~\eqref{eq::new-sum}, following the method discussed in Lemma~\ref{lem::particular_sol}, we can show that  $\vectsf{x}^{\text{\rm{p}}}(\vect{\delta}(t),\vect{u}(t),\vectsf{d}(t))=\sum_{i=1}^{n-1}(\delta_i(t) + u_i(t)+\mathsf{d}_i(t))\, \Bvectsf{x}^{\text{p},v_i}$ is a particular solution of~\eqref{eq::flow-aggregate-case2}, i.e., $\vectfrak{I}^{\text{o}}\,\vectsf{x}^{\text{\rm{p}}}(\vect{\delta}(t),\vect{u}(t),\vectsf{d}(t))=\vect{\delta}(t)+\vect{u}(t)+\vectsf{d}(t).$
 Therefore, for any given load vector $\vectsf{d}$ (to simplify the notation we drop argument $t$), we have
  \begin{align*}
    &\big\{ \vect{x}\in\real^m \mid 
\vectfrak{I}^{\text{o}}\vect{x}=\vect{\delta}+\vect{u}+\vectsf{d}, \vect{\delta}\in\real^n,\vect{u}\in\real^n \} = \nonumber\\
    &\Big\{ \vectsf{B}^\text{\rm{of}}\,^\top\vect{z}+\vectsf{x}^{\text{\rm{p}}}(\vect{\delta},\vect{u},\vectsf{d})\,\Big|\, \vectsf{x}^{\text{\rm{p}}}(\vect{\delta},\vect{u},\vectsf{d})\!=\!\!\sum\nolimits_{i=1}^{n-1}\!\!(\delta_i\! +\! u_i\!+\!\mathsf{d}_i)\, \Bvectsf{x}^{\text{p},v_i},\\
&~~\sum\nolimits_{i=1}^{n}\!\!(\delta_i + u_i+\mathsf{d}_i)=0,\,\vect{z} \in \mathbb{R}^{n-m+1},\vect{\delta}\in\real^n,\vect{u}\in\real^n
\Big\}.
\end{align*}
As a result, at each $t\in\mathcal{T}$, we can eliminate the affine equation~\eqref{eq::flow_case2} and arrive at the equivalent optimization problem~\eqref{eq::opt-pow-flow-equiv} whose minimizers are related to minimizers of~\eqref{eq::opt-pow-flow} in a way that is described in the statement.
\end{proof}

\section{A cycle-basis distributed ADMM algorithm for minimum cost network flow problem}\label{sec::cycle-ADMM}

In this section, we consider the minimum cost network flow problem~\eqref{eq::prob_def1} and show that its equivalent form~\eqref{eq::affine-eliminated-prob_def} based on the cycle basis variables is amenable to a distributed solution.  
We start by introducing some notation related to the oriented fundamental cycles of  $\GG^\text{o}$. Let $\mathcal{C}^\text{of}=\{\mathfrak{C}^\text{of}_i\}_{i=1}^{\mu}$, where $\mu=m-n+1$, be the set of fundamental cycles of 
$\GG^\text{o}$ whose cycle matrix $\vectsf{B}^\text{of}$ is used to eliminate the flow conservation equation as explained in Section~\ref{sec::prob_def}.  We represent the set of arcs of any $\mathfrak{C}^\text{of}_i\in\mathcal{C}^\text{of}$, $i\until{\mu}$, by $ \mathcal{E}^\mathfrak{C}_i=\{e_j\in\EE^\text{o},j\until{m}\,|\,\vectsf{B}^{\text{of}}_{ij}\neq 0\}$. For a given cycle basis $\mathcal{C}^\text{of}$, we refer to the cycles that share an arc   as neighbors and represent the set of (cycle) neighbors of any fundamental cycle $\mathfrak{C}^\text{of}_i\in\mathcal{C}^\text{of}$, $i\until{\mu}$, by $\mathcal{N}^{\mathfrak{C}}_i=\{j\until{\mu}\backslash\{i\}\,|\, \exists\, k\until{m}\text{~s.t.~}\vectsf{B}^\text{of}_{ik}\neq 0 \text{~and~} \vectsf{B}^\text{of}_{jk}\neq 0 \}$.
We let $\mathcal{C}^{\text{of}}(e_i)$ 
be the set of indexes of the fundamental cycles that arc $e_i\in\EE^\text{\text{o}}$ belongs to them, i.e., $\mathcal{C}^{\text{of}}(e_i)=\{j\until{\mu}\,|\, e_i\in\mathcal{E}^\mathfrak{C}_j\}$. 

Recall that to eliminate the flow conservation equation we used $\vect{x}=\vectsf{B}^\text{\rm{of}}\,^\top\vect{z}+\vectsf{x}^{\text{p}}$, or equivalently $x_i=\vect{z}^\top[\vectsf{B}^\text{\rm{of}}]_i+\mathsf{x}_i^{\text{p}}$, $i\until{m}$. Notice that one can think of every $z_i$, $i\until{\mu}$ as a cycle flow variable (with positive direction in counterclockwise direction) of the fundamental cycle $\mathfrak{C}^\text{of}_i$. Recall that, for a given arc $e_i\in\EE^\text{o}_{\text{physic}}$, every element of $\vectsf{B}^\text{of}_{ji}$ is zero except if cycle $\mathfrak{C}^\text{of}_j$ contains arc $e_i$, i.e., $e_i\in{\mathcal{E}}^\mathfrak{C}_j$. As a result, we can deduce that every $x_i$, $i\until{m}$, is an affine function of $\mathsf{x}_i^{\text{p}}$ and $\{z_k\}_{k\in\mathcal{C}^{\text{of}}(e_i)}$, indicating that every arc flow is a function of its particular solution and the cycle flow of fundamental cycles that contain the arc. Given such relationship, then the cost function of every arc as $\phi_i(x_i)=\phi_i(\vect{z}^\top[\vectsf{B}^\text{\rm{of}}]_i+\mathsf{x}_i^{\text{p}})=\psi_i(\{z_k\}_{k\in\mathcal{C}^{\text{of}}(e_i)})$. 

\emph{Cyber layer architecture:} based on the observation above, we propose to assign a cyber-layer node to each fundamental cycle (see Fig.~\ref{fig::CPS} as an example).  We assume that the cyber-layer nodes of neighboring fundamental cycles can communicate with each other in bi-directional way. For bi-connected physical layer graphs this procedure will result in a connected graph of $\mu$ nodes for cyber layer (see Fig.~\ref{fig::CPS} and Fig.~\ref{fig::min-cycle-numeric-example} for examples). To obtain fundamental cycles with fewest number of arcs we propose to use minimum weight cycle basis algorithms to generate the fundamental cycle basis for the cyber layer (see Appendix).

\begin{figure}[t]
  \centering
      \includegraphics[height=2.7in]{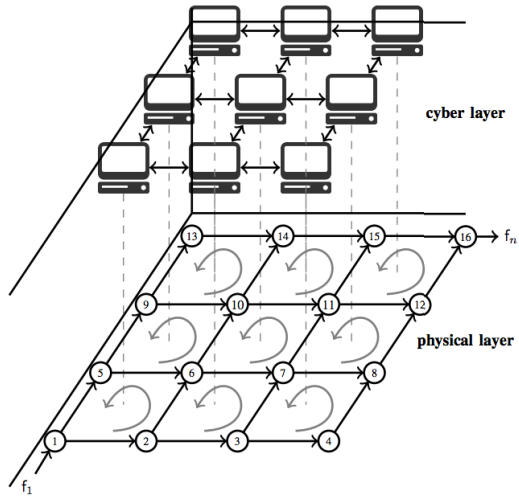}
      \caption{Physical and cyber layers of an optimal network flow problem: Physical layer has $n=16$ nodes and $m=24$. In this layer, the arrows indicate the positive flow directions. The cyber layer is constructed using the minimum wight cycle basis of the  physical layer graph. The cyber layer has $N=9$ agents with processing and communication capabilities.}\label{fig::CPS}\vspace{-0.1in}
\end{figure}
\subsection{Cycle Basis distributed ADMM solution for minimum cost network flow problem}
In this section, we derive an equivalent representation of optimization algorithm~\eqref{eq::affine-eliminated-prob_def} which can be solved in a distributed manner using the ADMM algorithm of~\cite{mota2015distributed} by our cycle-based cyber layer. 
To this end, for every cyber-layer node $i\until{\mu}$, we define $\vect{y}_i=(\bar{y}_i,\Tvect{y}_i)\in\real^{|\mathcal{N}^{\mathfrak{C}}_i|+1}$, where $\bar{y}_i\in\real$ is the local copy of $z_i$ and $\Tvect{y}_i$ is the local copy of $\{z_k\}_{k\in\mathcal{N}^{\mathfrak{C}}_i}$ at cyber node $i$.
With this definition, we assume that every cyber node, besides its own corresponding cycle flow, has also a copy of cycle flow variable of its neighbors. Next, we cast the cost function of each cyber node in terms of its decision variable $\vect{y}_i$. Let $\vect{y}_i(e_k)$ be the component(s) of $\vect{y}_i$ corresponding to $\{z_j\}_{\{j\in\mathcal{C}^{\text{of}}(e_k)\}}$. For every cyber-layer node, we define its cost function as
\begin{align}
\theta_i(\vect{y}_i)\,=\,\,&\sum_{\forall e_k\in\mathcal{E}^\mathfrak{C}_i} \frac{1}{|\mathcal{C}^{\text{of}}(e_i)|}\psi_k(\vect{y}_i(e_k)).
\end{align}
Then, we can cast the minimum cost network flow problem~\eqref{eq::affine-eliminated-prob_def} in the following equivalent form
\begin{align}\label{eq::optimal-network-flow-ADMM}
  & \vect{y}^\star=\underset{\vect{y}_1,\cdots,\vect{y}_{\mu}}{\argmin} 
   \sum\nolimits_{i=1}^{\mu}\theta_i(\vect{y}_i),\quad \text{~s.t.}\\
   &\text{set of constraints at each cyber agent ~}i\until{\mu}:\nonumber\\
&\begin{cases}
\vect{y}_i(e_k)=[\{\bar{y}_j\}_{j\in\mathcal{C}^{\text{of}}(e_i)}],\,\nonumber\\
\mathsf{b}_k\leq \vect{y}_i(e_k)^\top[\{\vectsf{B}^\text{of}_{jk}\}_{j\in\mathcal{C}^{\text{of}}(e_k)}]+\mathsf{x}^\text{p}_k\leq \mathsf{c}_k,
\end{cases}~\forall e_k\in\mathcal{E}^\mathfrak{C}_i.
\end{align}
In this formulation, every cycle-based cyber node has a copy of the cycle flows that go through its arcs, i.e, $\vect{y}_i$. The equality constraint at each node ensures that local copies of the cycle flows of the neighboring agents are the same, while the inequality constraint  ensures that the flow through the arcs' of each cycle respect the capacity bounds.

The new formulation~\eqref{eq::optimal-network-flow-ADMM} fits the standard framework developed for distributed ADMM solutions and can be solved for example using the algorithm of~\cite{mota2015distributed}. The details are omitted for brevity. In this distributed implementation we assume that elementary particular solution set $\{\Bvectsf{x}^{\text{p},v_i}\}_{i=1}^{n-1}$ are computed off-line and are available at cyber nodes. At operation times, we only need to broadcast the input/output flow vector $\vectsf{f}$ to the cyber-layer agents. For networks with large fundamental cycle sizes, one can split a cycle among several cyber nodes. In this case the length of the $\bar{y}_i$ of these agents will be $0$ and we can still use the distributed ADMM algorithm  to solve the problem. Similarly, two or more cycles can be assigned to one cyber layer.

\subsection{Numerical Example}\label{sec::numeric}
\begin{figure}[t]
\centering
 \includegraphics[height=2.7in]{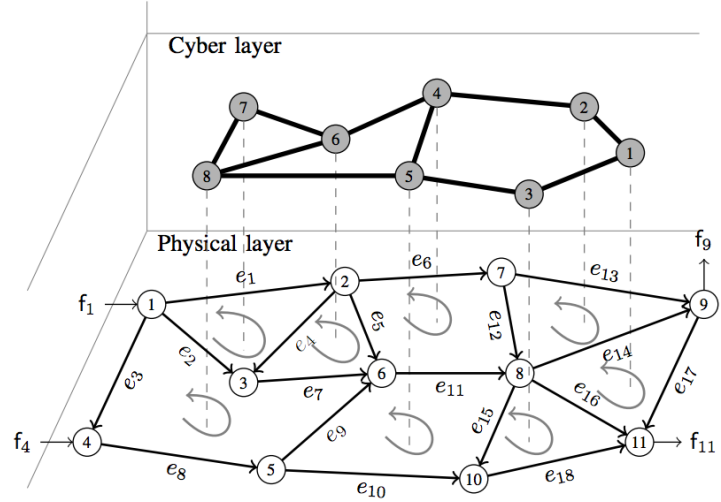}
\caption{A physical network with a cycle-based cyber layer network overlaid atop . The physical layer network has two source nodes $v_1$ and $v_4$ and two sink nodes $v_9$ and~$v_{11}$.
}\label{fig::min-cycle-numeric-example}
\end{figure}
In this section, we demonstrate the use of distributed cycle-based ADMM algorithm for a minimum cost optimal flow over the network shown in Fig.~\ref{fig::min-cycle-example}.
 We assign positive flow orientation to the arcs as represented in Fig.~\ref{fig::min-cycle-numeric-example}. We generate the cyber layer based on the minimum weight cycle basis as shown in the bolder network with gray nodes in Fig.~\ref{fig::min-cycle-example}. In this problem, we set $\mathsf{b}_i= -\,\mathsf{c}_i$, and  $\mathsf{c_i}\in\realpositive$, $i\until{18}$. We assume that the cost of the network flow at each arc is given as $\phi_i(x_i)=(\frac{x_i}{\mathsf{c}_i})^2$, where $x_{i} = \vect{z}^\top[\vectsf{B}^\text{\rm{of}}]_i+\mathsf{x}_i^{\text{p}}$ is the arc flow and $\vect{z} = (z_1,\cdots,  z_8)^\top$ are cycle flows. In the physical layer network in Fig.~\ref{fig::min-cycle-numeric-example}, there are two source nodes $v_1$ and $v_4$ and two two sink nodes $v_9$ and $v_{11}$.

 We follow Lemma~\ref{lem::particular_sol} to generate a particular solution for given input output flow vector $\vectsf{f}=(\mathsf{f}_1,0,0,\mathsf{f}_4,0,0,0,0,\mathsf{f}_9,0,-(\mathsf{f}_1+\mathsf{f}_4+\mathsf{f}_9))^\top$-- recall that $\mathsf{f}_{11}=-(\mathsf{f}_1+\mathsf{f}_4+\mathsf{f}_9)$ (recall that input flows have positive and output flows have negative values). We compute the elementary particular solutions for nodes $v_1$, $v_4$ and $v_9$ using shortest path from them to node $v_{11}$: $\Bvectsf{x}^{\text{p},v_1}=(0,0,1,\vect{0}_{1\times 4},1,0,1,\vect{0}_{1\times 7},1)^\top$, $\Bvectsf{x}^{\text{p},v_4}=(\vect{0}_{1\times 7},1,0,1,\vect{0}_{1\times 7},1)^\top$, and $\Bvectsf{x}^{\text{p},v_{9}}=(\vect{0}_{1\times 16},1,0)^\top$. Then $\Bvectsf{x}^{\text{p}}=\mathsf{f}_1\,\Bvectsf{x}^{\text{p},v_1}+\mathsf{f}_4\,\Bvectsf{x}^{\text{p},v_4}+\mathsf{f}_9\,\Bvectsf{x}^{\text{p}}$.

 In our simulation, the lower and upper capacity bounds are selected uniformly randomly from $[2,50]$, i.e.,  $\mathsf{c}_i\in[2,50]$, $i\until{m}$. Specifically, at arcs connected to sink and source points we have
 $\mathsf{c}_1=30$, $\mathsf{c}_2=2$, $\mathsf{c}_3=27$, $\mathsf{c}_8=50$, $\mathsf{c}_{13}=49$, $\mathsf{c}_{14}=23$, $\mathsf{c}_{15}=21$, $\mathsf{c}_{16}=11$, $\mathsf{c}_{17}=16$, and $\mathsf{c}_{18}=37$. For our selected capacity bounds, using
 Edmonds-Krap algorithm~\cite{edmonds1972theoretical} we obtain the maximum input flow $\mathsf{f}_1+\mathsf{f}_4$ to be $82$. In our simulation, then we set 
 $\mathsf{f}_1=32$ and $\mathsf{f}_4=50$. The output follows are $\mathsf{f}_9=-52$  and $\mathsf{f}_{11}=-30$. 
  After, $50$ iteration the input/output flows change to 
 $\mathsf{f}_1=15$, $\mathsf{f}_4=30$, $\mathsf{f}_9=-15$, and $\mathsf{f}_{11}=-30$.

\begin{figure}[t!]%[htbp]
  \centering   
  \subfloat[cyber-layer node 1 (fundamental cycle 1)]{
    \includegraphics[height=1.2in]{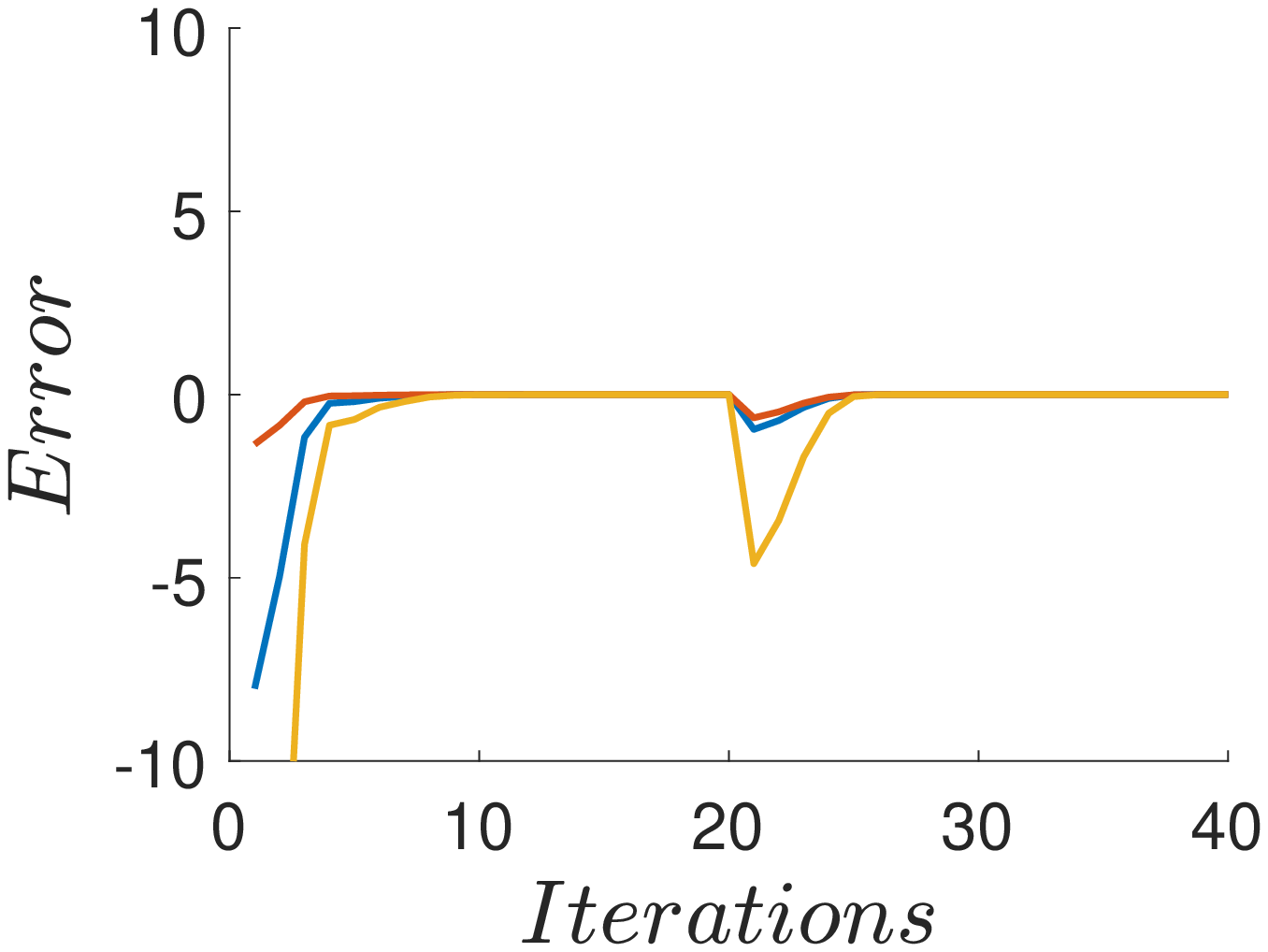}
  }~~
  \subfloat[cyber-layer node 2 (fundamental cycle 2)]
  {
    \includegraphics[height=1.2in]{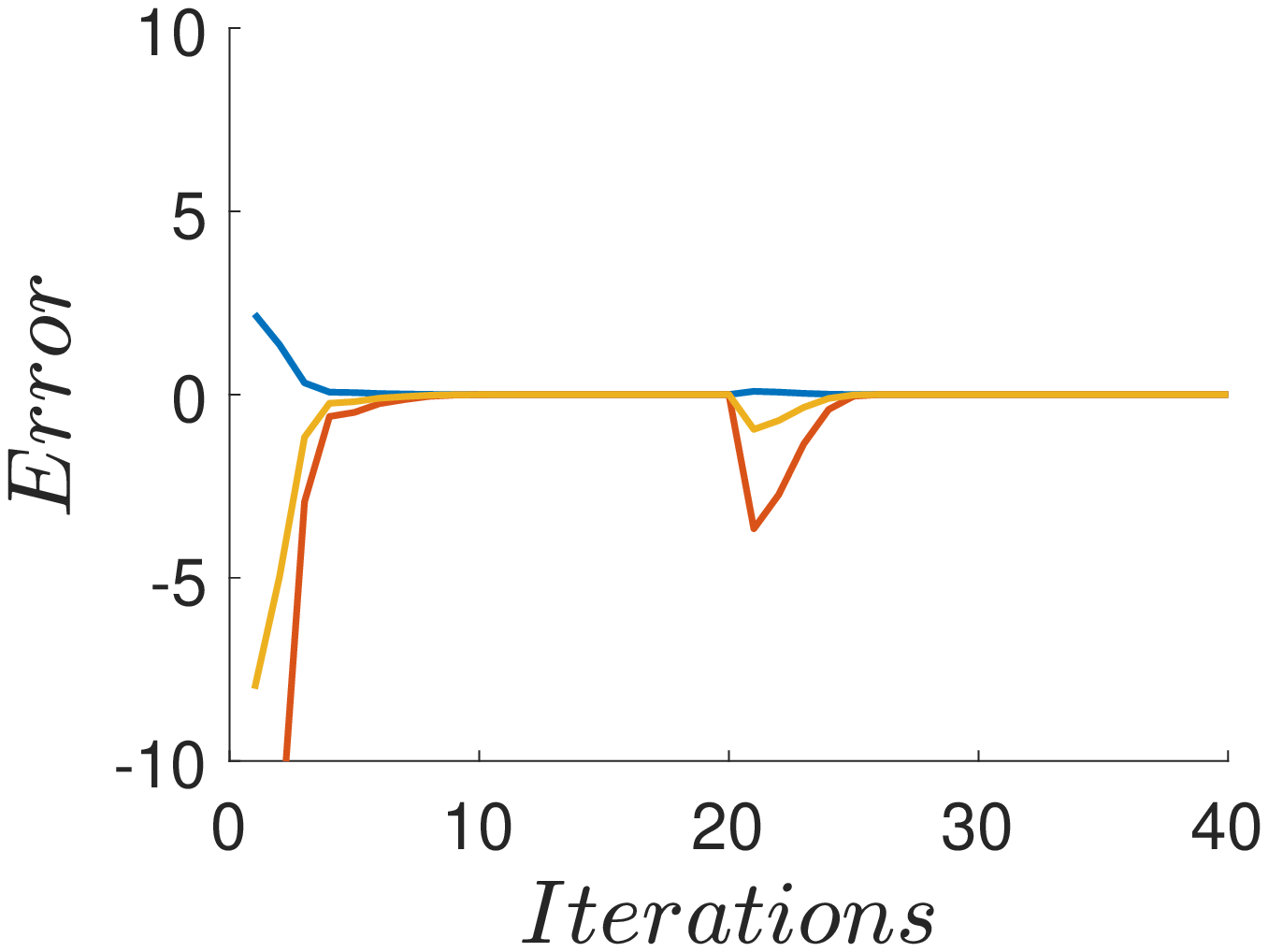}
  }\\ %goes to next line
  \subfloat[cyber-layer node 3 (fundamental cycle 3)]{
    \includegraphics[height=1.2in]{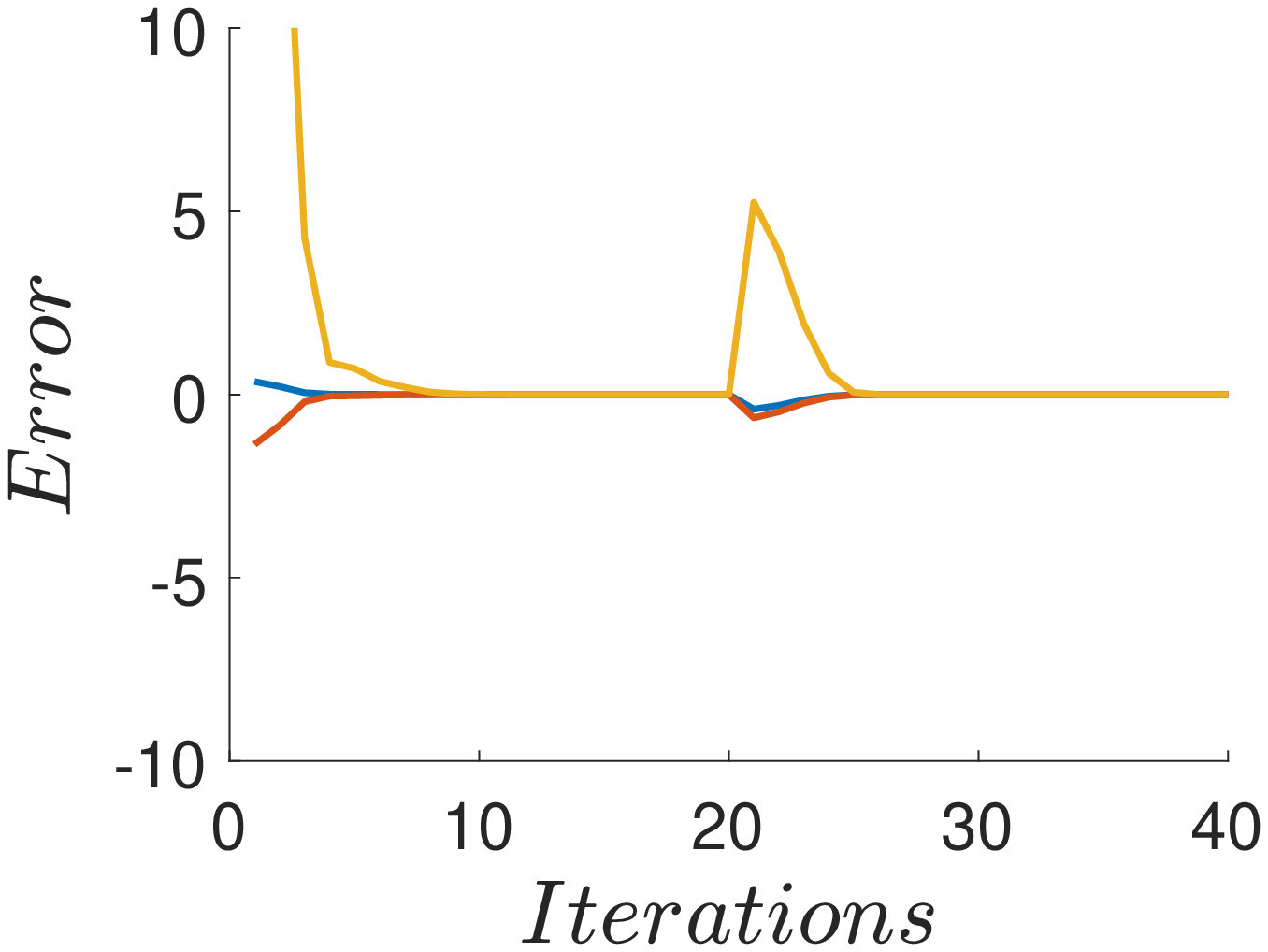}
  }~~
  \subfloat[cyber-layer node 4 (fundamental cycle 4)]
  {
    \includegraphics[height=1.2in]{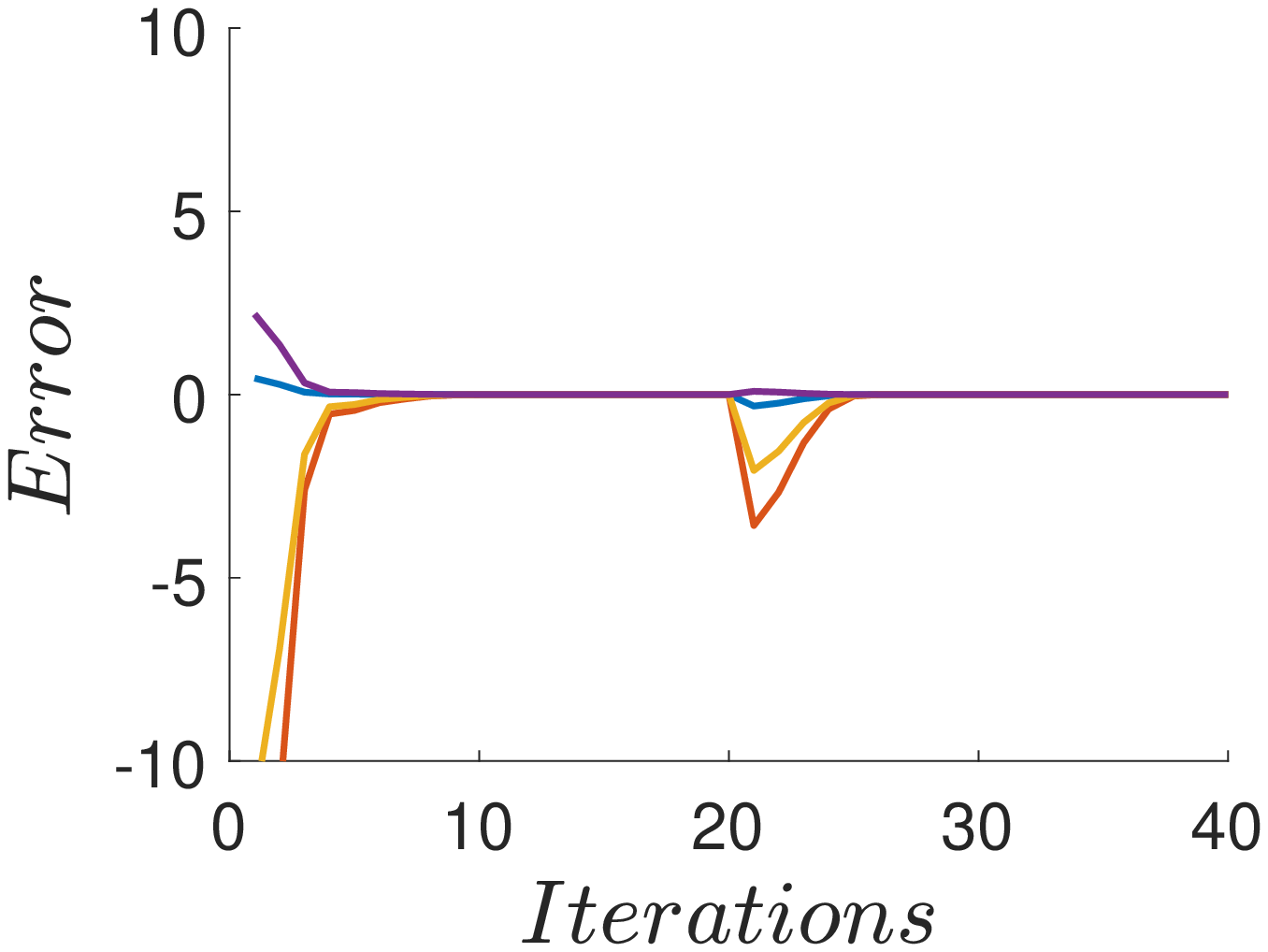}
  }\\
  \subfloat[cyber-layer node 5 (fundamental cycle 5)]
  {
    \includegraphics[height=1.2in]{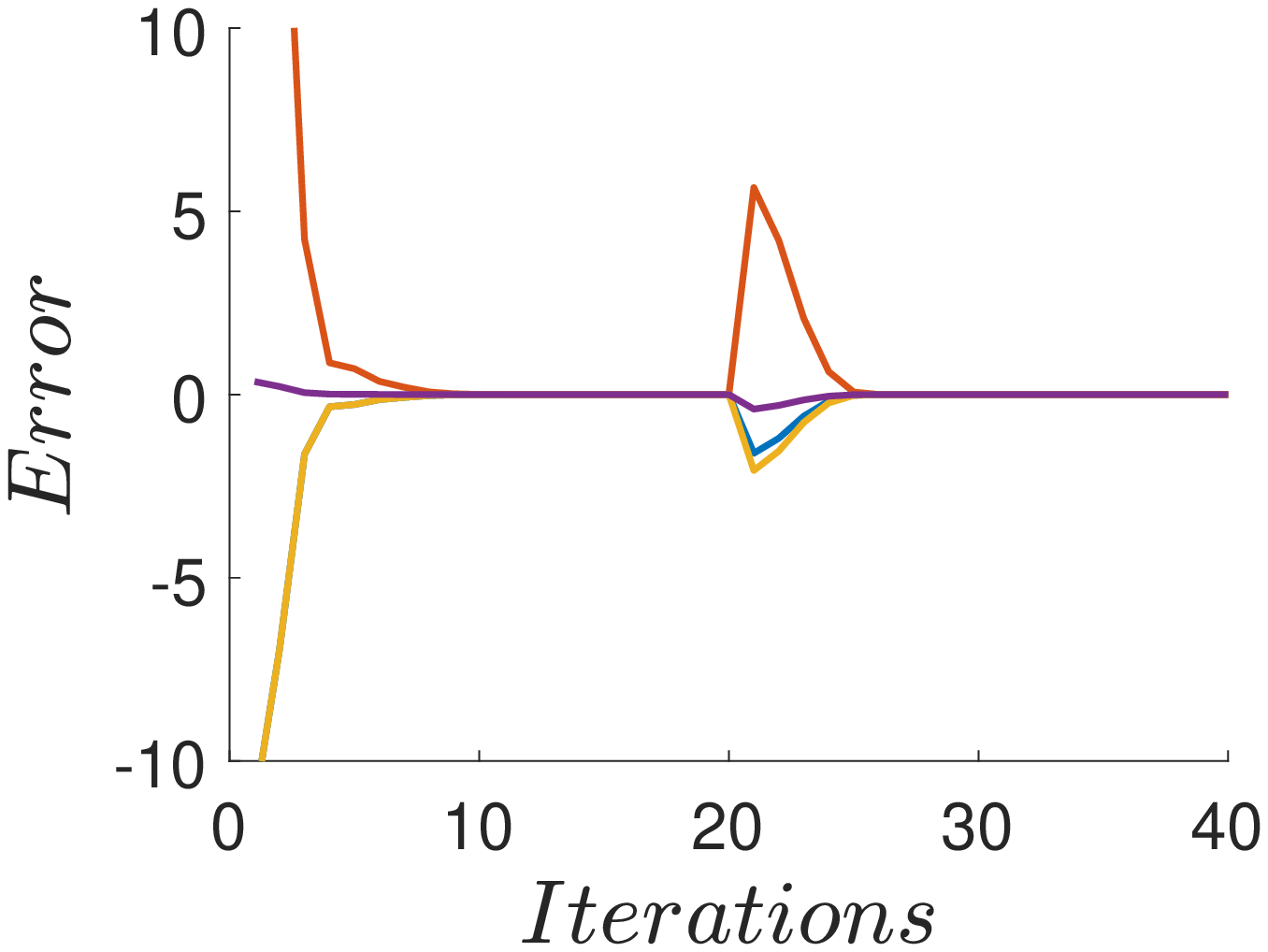}
  }~~
  \subfloat[cyber-layer node 6 (fundamental cycle 6)]
  {
    \includegraphics[height=1.2in]{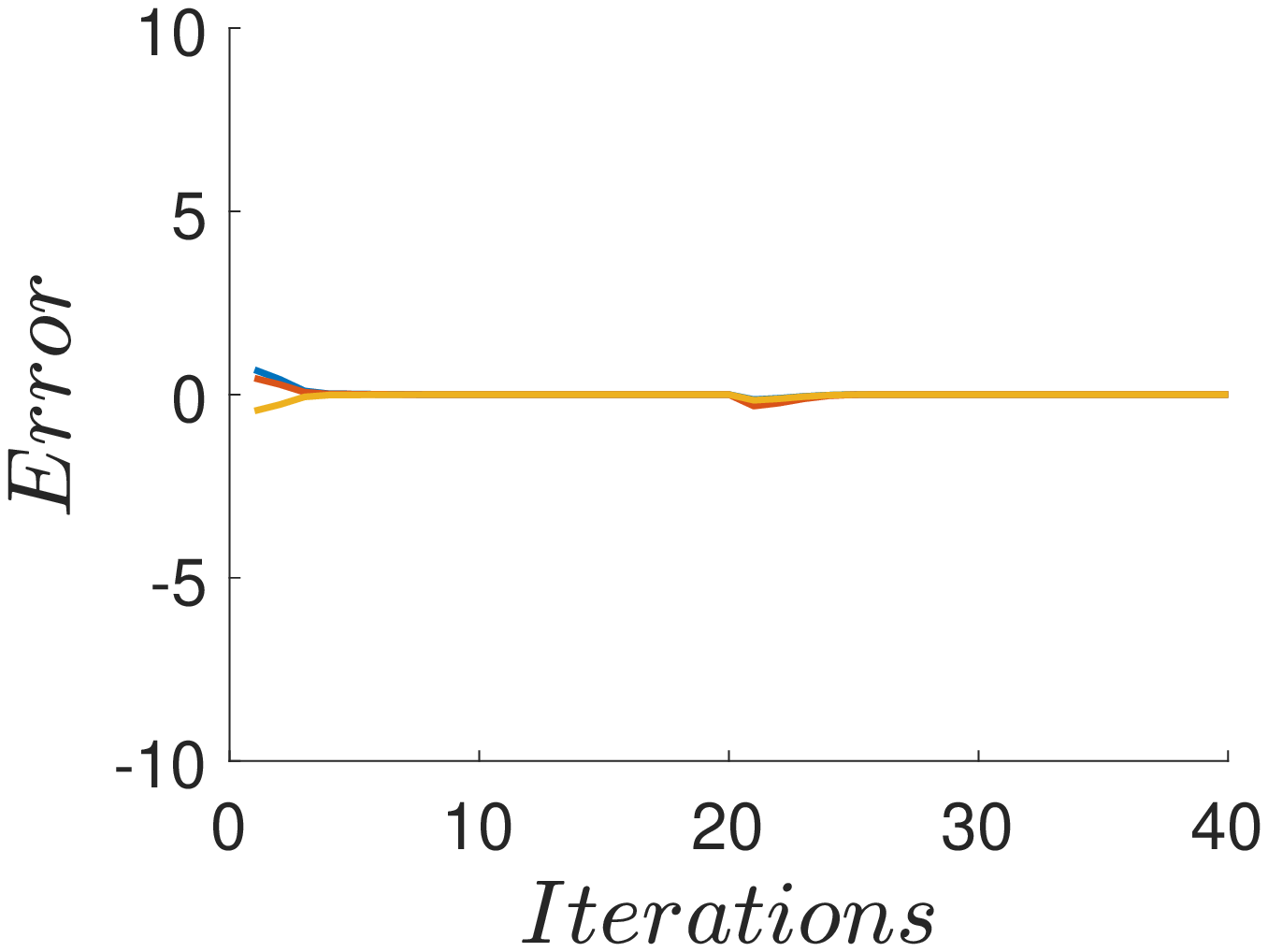}
  }\\
  \subfloat[cyber-layer node 7 (fundamental cycle 7)]
  {
    \includegraphics[height=1.2in]{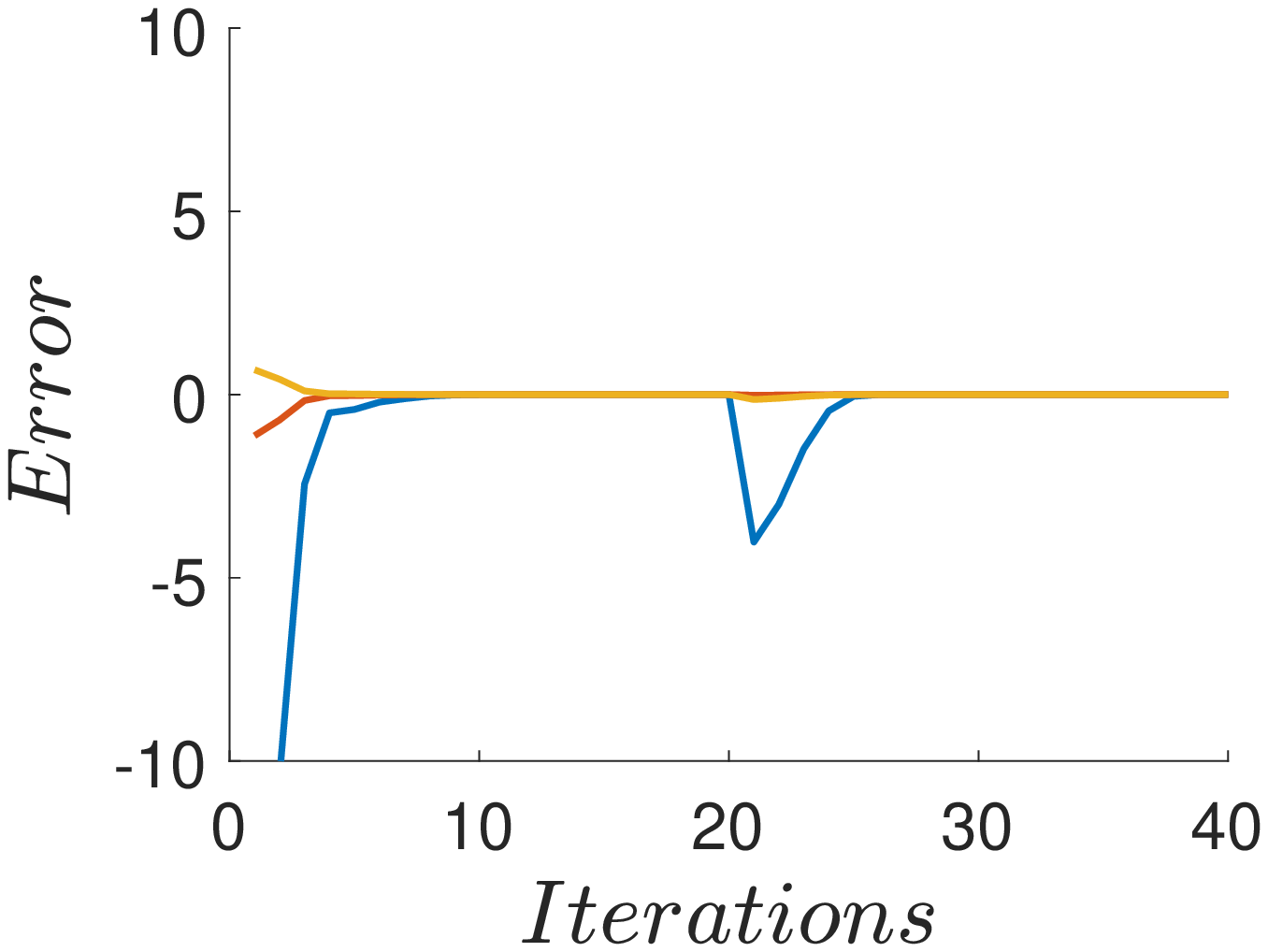}
  }~~
  \subfloat[cyber-layer node 8 (fundamental cycle 8)]
  {
    \includegraphics[height=1.2in]{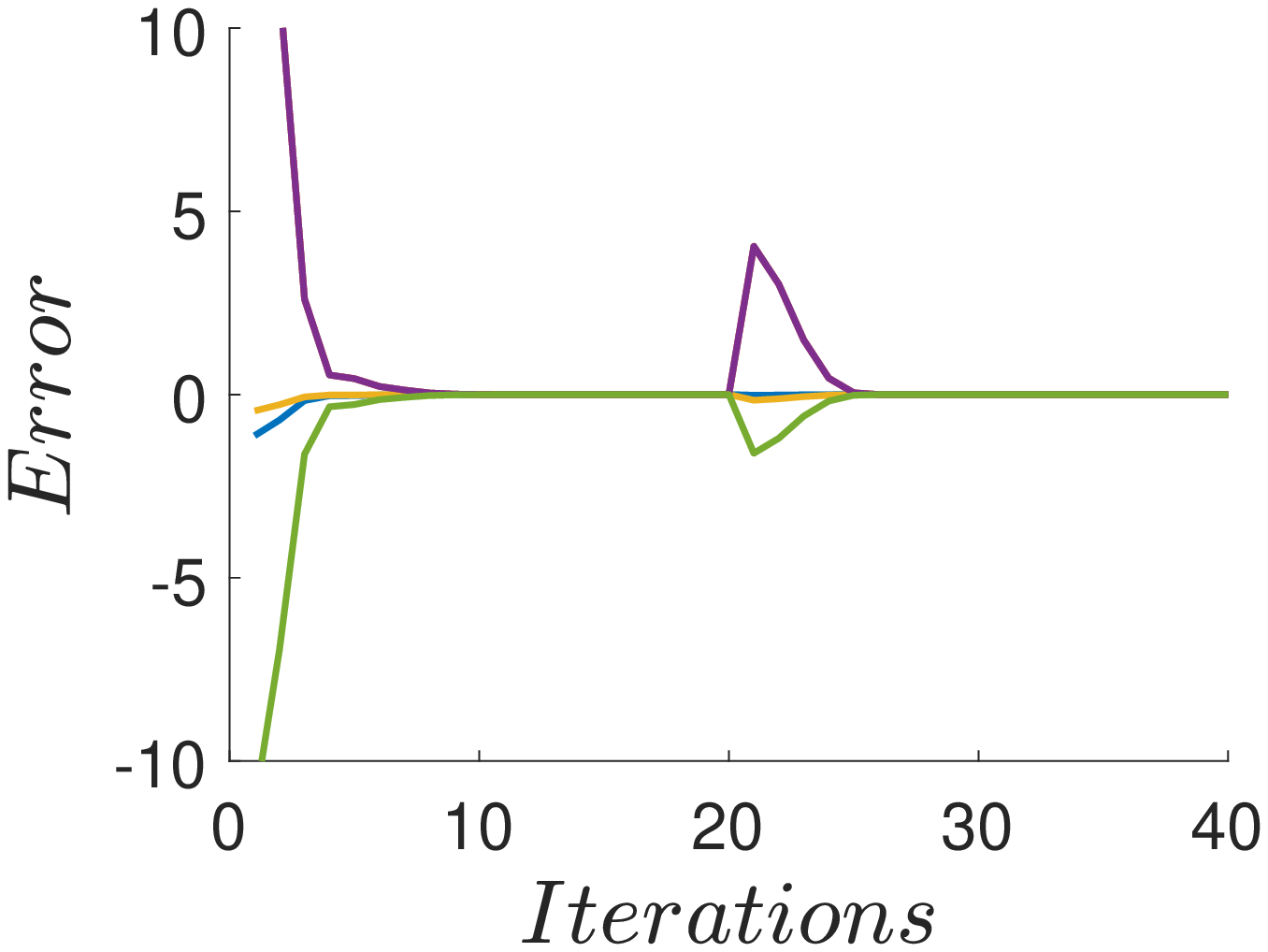}
  }
  \caption{Each plot depicts $x_i(k)-x_i^\star$, for $e_i$'s in that sub-captioned fundamental cycle. As this figure shows, every cyber node asymptotically calculates the optimal arc flow for its arcs.}\label{fig::simulation_results}\vspace{-0.1in}
\end{figure}

We use Matlab `quadprog' to solve the problem in a centralized manner to generate reference values to compare the performance of our distributed cycle basis distributed ADMM algorithm as outlined in Section~\ref{sec::cycle-ADMM}. The results are depicted in Fig.~\ref{fig::simulation_results}. In Fig.~\ref{fig::simulation_results}, plots show the distance of arc flow values from their optimum solution during execution of distributed ADMM. During the first 50 iteration the distributed ADMM converges to the optimum solution. Then, for the second external flows, it converges to the optimum solution again.

 \begin{comment}
\margin{we may not need this}\solmaz{For a given network and capacity bounds, maximum (resp. minimum) network flow problem gives an upper and (resp. lower) bound on the admissible ranges of input flow $f_{\text{in}}$ such that the feasible set~\eqref{eq::feasible-set} is always non-empty. Maximum (also minimum) flow of a network can be find using Edmonds-Krap algorithm in $\mathcal{O}(nm^{2})$ in a central way~\cite{edmonds1972theoretical}.}
\end{comment}

\section{Conclusion and future works}\label{sec::conclude}
We considered optimal network flow problems and investigated how the decision variables of these problems can be reduced  by eliminating the affine flow conservation equations. Our study was based on exploiting cycle basis concept from graph theory to eliminate flow conservation equation in an efficient manner. In particular, we showed that the computation regarding the proposed variable reduction can be done in a systematic manner, in polynomial time, using existing algorithms. Moreover, we showed that the new formulation of the optimal network flow problems with reduced variables is amenable to distributed solvers. In this regard, we constructed a cyber-layer structure based on cycles in the physical-layer network. We also demonstrated the use of a distributed ADMM solver  for minimum cost flow problem.

\bibliographystyle{ieeetr}%

%%%%%%%%%%%%%%%%

\appendix
\emph{Minimum weight cycle basis problem} is defined as the problem of finding an unoriented fundamental cycle matrix in which the total length of cycles is minimum. For graphs with positive arc weights, a solution can be found in polynomial time ~\cite{horton1987polynomial}. Here, our graph arc weights are $0$ and $1$. This algorithm generates a set of fundamental cycles, but restricts the generated cycles to a small set of $\mathcal{O}(nm)$ cycles, called Horton cycles. Each shortest path tree of the given graph has a set of fundamental cycles. The bound on Horton cycle is defined as $m$ fundamental cycles of $n$ shortest path trees. Every cycle in the minimum weight cycle basis is a Horton cycle~\cite{horton1987polynomial}. Dijkstra's algorithm finds $n$ shortest path trees and the Gaussian elimination is used to find independent cycles on an increasing ordered set of cycles.
Figures~\ref{fig::min-cycle-example},~\ref{fig::CPS} and~\ref{fig::min-cycle-numeric-example} depict  graphs with their minimum weight cycle basis highlighted. Improvement for worst case time complexity of this algorithm was provided for undirected graphs~\cite{amaldi2010efficient}, and planar graphs~\cite{borradaile2015min}.
\end{document}